\documentclass[11pt, lettersize]{article}
\usepackage{geometry}[margin = 1in]
\usepackage[utf8]{inputenc}
\usepackage{amsmath}
\usepackage{amsfonts}
\usepackage{amsthm}
\usepackage{graphicx}
\usepackage{wrapfig}
\usepackage{makecell}
\usepackage{xspace}
\usepackage{caption}
\usepackage{imakeidx}
\usepackage{thmtools}
\usepackage{enumerate}
\usepackage{xcolor}
\usepackage{lineno}
\usepackage[shortlabels]{enumitem}
\usepackage{hyperref}
\usepackage{todonotes}

\newcommand{\e}{\ensuremath{e}}

\newcommand{\F}{\ensuremath{\mathcal{F}}}
\newcommand{\T}{\ensuremath{\mathcal{T}}}
\newcommand{\E}{\ensuremath{\mathcal{E}}}
\newcommand{\etal}{\textit{et al.}\xspace}
\newcommand{\mypar}[1]{\smallskip\noindent{\bfseries\boldmath#1}}

\newcommand{\tor}{\!\to\!}

\newtheorem{property}{Property}

\newtheorem{theorem}{Theorem}

\newtheorem{lemma}{Lemma}

\newtheorem{observation}{Observation}

\title{Dynamic Embeddings of Dynamic Single-Source Upward Planar Graphs}

\author{Ivor van der Hoog\thanks{
    Algorithms, Logic and Graphs. Technical University of Denmark, Denmark. \texttt{idjva@dtu.dk}.
  }, Irene Parada\thanks{Geometric Computing, Utrecht University, Netherlands, \texttt{	i.m.deparada@uu.nl}.}, Eva Rotenberg\thanks{Algorithms, Logic and Graphs. Technical University of Denmark, Denmark. \texttt{erot@dtu.dk}.}}

 \date{}
\begin{document}

%
\thispagestyle{empty}
\maketitle

\setcounter{page}{0}
\begin{abstract}
A directed graph $G$ is \emph{upward planar} if it admits a planar embedding such that each edge is $y$-monotone. Unlike planarity testing, upward planarity testing is NP-hard except in restricted cases, such as when the graph has the single-source property (i.e. each connected component only has one source).

In this paper, we present a dynamic algorithm for maintaining a combinatorial embedding $\mathcal{E}(G)$ of a single-source upward planar graph subject to edge deletions, edge contractions, edge insertions upwards across a face, and single-source-preserving vertex splits through specified corners.
We furthermore support changes to the embedding $\mathcal{E}(G)$ on the form of subgraph flips that mirror or slide the placement of a subgraph that is connected to the rest of the graph via at most two vertices. 

All update operations are supported as long as the graph remains upward planar, and all queries are supported as long as the graph remains single-source.
Updates that violate upward planarity are identified as such and rejected by our update algorithm.
We dynamically maintain a linear-size data structure on $G$ which supports incidence queries between a vertex and a face, and upward-linkability of vertex pairs.
If a pair of vertices are not upwards-linkable, we facilitate one-flip-linkable queries that point to a subgraph flip that makes them linkable, if any such flip exists. 

We support all updates and queries in $O(\log^2 n)$ time.
\end{abstract}

\thispagestyle{empty}
\newpage
\section{Introduction}

When visualizing acyclic directed graphs, it is desirable to display the underlying hierarchy. 
A directed graph is \emph{upward planar} if it admits a drawing that is both \emph{upward} (every edge has  monotonically increasing $y$-coordinates) and \emph{planar} (without crossings). 
Upward planar drawings highlight the hierarchy of the directed graph and upward planarity is a natural analogy of planarity for directed graphs. 
We can test if a graph admits a planar embedding in linear time since 1974~\cite{HopcroftT74planarity}.
In sharp contrast, testing upward planarity is, in general, NP-complete~\cite{GargT01testing}. 

In this work we consider a dynamic $n$-vertex single-source directed graph $G$ with an upward embedding $\E(G)$ and design a data structure that supports updates to the graph (that preserve upward planarity and $G$ being single-source), flips in the embedding as well as uplinkability and one-flip uplinkability queries in $O(\log n^2)$ time. 
We remark that the bounds we obtain in this more challenging upward-planar setting match the ones of the planar case~\cite{holm2017dynamic} 
and generalize the dynamic data structure for single-source directed graphs with a fixed embedding and outer face~\cite{RextinH17dynamicup}.
In particular, we allow not only changes to the embedding, we allow changes to the outer face including the out-edges of the source.

For planar graphs, a corresponding dynamic data structure that supports similar queries and changes to the embedding~\cite{holm2017dynamic} has been proven instrumental for (fully) dynamic planarity testing of a dynamic graph $G$~\cite{DBLP:conf/soda/HolmR20}. 
An additional motivation for this work is the hope that with additional insights and techniques, it would serve as a stepping stone towards a fully-dynamic algorithm for upward planar graphs.



\mypar{Upward planarity.}
Upward planar graphs are spanning subgraphs of planar $st$-graphs~\cite{BattistaT88stgraphs,Kelly87stgraphs,Platt76stgraphs}. 
While testing upward planarity is NP-complete for general graphs~\cite{GargT01testing}, 
polynomial algorithms have been devised for several classes of directed graphs. 
One of the most relevant such results is the linear algorithm for single-source directed graphs~\cite{BertolazziBMT98optimal,bruckner2019spqr,HuttonL96source}. 
Other classes for which upward planarity can be tested in polynomial time are 
graphs with a fixed (upward) embedding~\cite{BertolazziBLM94triconnected},  
outerplanar graphs~\cite{Papakostas94outerplanar}, 
and series–parallel graphs~\cite{DidimoGL09SPgraphs}. 
Upward planarity testing is fixed-parameter tractable in the number of triconnected components and cut vertices~\cite{Chan04fpt,HealyL06fpt}, and in the number of sources~\cite{ChaplickGFGRS22param}. 
Regarding dynamic algorithms, 
it can be checked in $O(\log n)$ amortized time whether an embedded single-source upward planar directed graph with a fixed external face can stay as such in the presence of edge insertions and edge deletions without changing the embedding~\cite{RextinH17dynamicup}. 
There are no existing worst-case results for dynamic upward embeddings and no results for dynamic upward embeddings subject to flips in the embedding.
The study of upward planar graphs continues to be a prolific area of research, including recent developments in 
parameterized algorithms for upward planarity~\cite{ChaplickGFGRS22param}, 
bounds on the page number~\cite{JungeblutMU22pagenr},
morphing~\cite{LozzoBFPR20morph}, 
and extension questions~\cite{bruckner2019spqr,LozzoBF20extend}. 
In particular, it can be tested in $O(n^2)$ time whether a given drawing can be extended to an upward planar drawing of an $n$-vertex single-source directed graph $G$~\cite{bruckner2019spqr}. 
For $st$-graphs the running time was recently improved to $O(n \log n)$~\cite{LozzoBF20extend}. 
This contrasts with the linear-time extension algorithm for the planar undirected case~\cite{AngeliniBFJKPR15extending}.

\mypar{Dynamic maintenance of planar (embedded) graphs.} 
Dynamic maintenance of graphs and their embedding is a well-studied topic in theoretical computer science \cite{DBLP:conf/icalp/BattistaT90,DBLP:journals/siamcomp/BattistaT96,DBLP:conf/soda/Eppstein03,DBLP:journals/jcss/EppsteinGIS96,DBLP:conf/stoc/GalilIS92,DBLP:conf/esa/HenzingerP95,holm2017dynamic,DBLP:conf/stoc/HolmR20,DBLP:conf/soda/HolmR20,DBLP:conf/stoc/Poutre94,DBLP:conf/icalp/Westbrook92}. 
In this area, we typically study some graph $G = (V, E)$ subject to adding or removing edges to the edge set $E$. 
A combinatorial embedding $\E(G)$ specifies the outer face and for each vertex $v \in V$ a cyclical ordering of the faces incident to  $v$.
One famous dynamic maintenance of an embedding is the work by Eppstein~\cite{DBLP:conf/soda/Eppstein03} who studies maintaining an embedding $\E(G)$ subject to edge deletions and insertions across a specified face, plus the genus of $\E(G)$ as a crossing free drawing, in $O(\log n)$ time. 
Henzinger, Italiano, and La Poutre give a dynamic algorithm for maintaining a plane embedded graph subject to edge-deletions and insertions across a face, while supporting queries to whether a pair of vertices presently share a face in the embedding, in $O(\log ^2 n)$ time per operation~\cite{DBLP:conf/esa/HenzingerP95}. 
Holm and Rotenberg~\cite{holm2017dynamic} expand on the result of \cite{DBLP:conf/esa/HenzingerP95} by additionally supporting 
\emph{flips} (operations which change the combinatorial embedding but not $G$). They show how to support all operations in $O(\log^2 n)$ time, rejecting operations that would violate the planarity of $\E(G)$.
In addition, their data structure supports \emph{linkability} queries which, for a pair of vertices, report a sequence of faces across which they are linkable, or singular flips in the embedding that would make the two vertices linkable.

For dynamic planarity testing, that is, maintaining a bit indicating whether the dynamic graph is presently planar, there has been a body of work 
\cite{DBLP:conf/icalp/BattistaT90,DBLP:journals/siamcomp/BattistaT96,DBLP:journals/jcss/EppsteinGIS96,DBLP:conf/stoc/GalilIS92, DBLP:conf/stoc/HolmR20,DBLP:conf/soda/HolmR20, DBLP:reference/algo/Patrascu08,DBLP:conf/stoc/PatrascuD04,DBLP:conf/stoc/Poutre94,DBLP:conf/icalp/Westbrook92}.
The current state-of-the-art algorithm for incremental planarity testing is by Holm and Rotenberg~\cite{DBLP:conf/soda/HolmR20} in $O(\log^3 n)$ worst case time per edge insertion.
Here, they crucially rely upon an $O(\log^2 n)$ fully dynamic algorithm for the dynamic maintenance of a planar embedding $\E(G)$~\cite{holm2017dynamic}.

\mypar{Contribution and organisation.}
Let $G$ be a single-source, upward planar directed graph and $\E(G)$ an upward planar combinatorial embedding.
Section~\ref{sec:preliminaries} contains our preliminaries: and formalizes our setting and the type of updates and queries that we allow.
In Section~\ref{sec:datastructure} we present our linear size dynamic data structure that maintains $G$, and a combinatorial representation of $\E(G)$ subject to an edge insertion across a face, edge deletion, vertex splitting, edge contraction, and `flip'-operations that perform local changes to the embedding,
in $O(\log ^2 n)$ time per operation. 
We illustrate for an edge insertion between two corners $c_u$ and $c_v$ in the embedding how to update our data structure and verify whether inserting the edge $c_u \tor c_v$ violates upward planarity.
In Section~\ref{sec:queries} we extend this data structure to support \emph{uplinkable queries}: i.e. queries to whether a directed edge may be inserted across a face in the current embedding such that the resulting combinatorial embedding remains upward planar. 
Specifically, our algorithm is capable of reporting all suitable faces (and corners) of the combinatorial embedding across which an edge from $u$ to $v$ may be inserted. In the negative case, Section~\ref{sec:queries} furthermore shows how to support queries to whether minor changes to the embedding would allow us to insert the edge; namely, the \emph{one-flip uplinkable} query which answers whether
there is a sequence of flips involving at most two faces after which 
$u$ and $v$ are uplinkable, and outputs such a sequence if one exists.


\newpage
\section{Preliminaries}
\label{sec:preliminaries}

Let $G = (V, E)$ be a directed graph (\emph{digraph}) where $V$ is a set of vertices and $E$ is a set of ordered pairs of vertices: each pair $(u, v)\in E$ representing an edge directed from $u$ to $v$. 
We introduce the concepts and data structures that will be used in the rest of the paper.

\mypar{(Upward) combinatorial embeddings.} 
A \emph{combinatorial embedding} $\E(G)$ of $G$ specifies for every vertex $v \in V$ a (counter-clockwise) cyclical ordering of the edges incident to~$v$. 
Note that this defines a set of faces $F$ where for every $f \in F$ there is a cyclical ordering of vertices incident to $f$, and for convenience we explicitly store this. 
An \emph{upward (combinatorial) embedding} additionally stores for every vertex $v \in V$ for any consecutive pair of edges whether their angle is reflex. 
An (upward combinatorial) embedding $\E(G)$ is upward planar whenever there exists an upward planar drawing $D(G)$ matching $\E(G)$. 

\mypar{Corners and the face-sink graph.}
In a digraph a \emph{source} (resp. a \emph{sink}) is a vertex with only incoming (resp. outgoing) edges; 
the other vertices are called \emph{internal}.
For any embedding, we define a \emph{corner} of the embedding for every 4-tuple $(v, f, e_1, e_2)$ where $v\in V$ is incident to a face $f$, $e_1, e_2 \in E$ are both incident to $v$ and $f$, and $e_1, e_2$ are consecutive in the counter cyclical order around $v$.\footnote{For the trivial graph consisting of a singleton vertex, there is a corner which is just the two-tuple $(v,f)$.}
A \emph{sink corner} is a corner where $e_1$ and $e_2$ are both directed towards $v$.
A \emph{top corner}
is a sink corner where the angle between $e_1$ and $e_2$ is convex.
If a vertex $v$ is incident to a sink corner of $f$ which is not a top corner, we refer to $v$ as a \emph{spike} of $f$.
A \emph{critical vertex} is an internal vertex
in~$G$ incident to at least one sink corner.
Bertolazzi \etal~\cite{BertolazziBMT98optimal} define for any (combinatorial) embedding $\E(G)$ of a single-source digraph $G$, its face-sink graph  $\F(\E(G))$. 
We slightly modify their definition: 
we define $\F(\E(G))$ as bipartite graph between the vertices $V$ and faces $F$ of $\E(G)$ where there is a directed edge from $v\in V$ to $f\in F$ whenever they share a spike, and a directed edge from $f$ to $v$ whenever they share a top corner. Bertolazzi \etal show:

\begin{theorem}[Theorem 1 in \cite{BertolazziBMT98optimal}, Fact 2 + 3]
\label{thm:facesink}
Let $G$ be a single-source digraph.
An upward embedding $\E(G)$ of $G$ is upward planar if and only if:
\begin{itemize}[noitemsep, nolistsep]
    \item $\F(\E(G))$ is a forest of trees $\T^*, \T_1, \T_2, \ldots \T_m$ where:
    \item $\T^*$ has as root the outer face of $\E(G)$ which is incident to the unique source of $G$,
    \item $\T_1, \ldots \T_m$ each have as root the unique critical vertex in $\T_i$. 
\end{itemize}
\end{theorem}

\mypar{Articulation slides, twists and separation flips.}
For any graph $G$ with embedding $\E(G)$, an \emph{articulation point} is a vertex $w$ such that its removal separates a connected component $G'$ of $G$ into at least two connected components $C_1,C_2,\ldots$: the \emph{articulation} components. 
In a planar graph, articulation points have two corners incident to the same face.
Given an articulation point $w$, let $C_w(v)$ be the component containing $v$ in $G\setminus w$. The following operations change the embedding while leaving the graph fixed.
An \emph{articulation slide} $F_s = (f, w, v)$ changes $\E(G)$ as follows:
we cut $C_w(v)$ from $w$ and embed it into $f$ merging $C_w(v)$ with $w$ (Figure~\ref{fig:flips}(a)). 
We refer to this operation as \emph{sliding} 
into $f$.

An \emph{articulation twist} $F_t = (w, v)$ changes $\E(G)$ by isolating the component $C_w(v)$ by cutting through $w$, mirroring its embedding, and merging it back with $w$ 
(Figure~\ref{fig:flips}(b)). 

A \emph{separation pair} in a graph $G$ is any pair of vertices whose removal separates a connected component into at least two components, and there exists some cycle whose vertices get separated into different components.
In a planar graph, $(x, y)$ are a separation pair if they are both incident to the same two faces.
The \emph{separation flip} $F = (f,g,x, y, v)$ modifies $\E(G)$ as follows: we split $x$ and $y$ along the corners incident to $f$ and $g$,
mirror the embedding of the component containing $v$,
and merge it back into $G$ 
(Figure~\ref{fig:flips}(c)).

\begin{figure}[tb]
\centering
\includegraphics{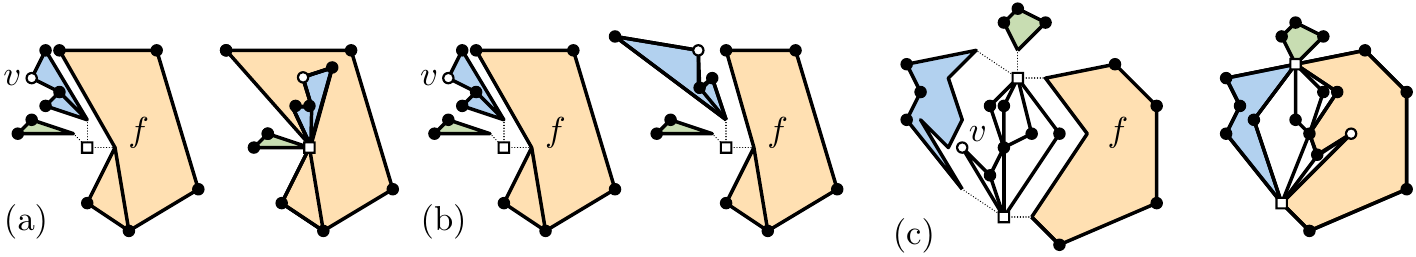}
\caption{
(a) An articulation vertex $w$ with $v$ in the blue component $B$.
$F_s = (f, w, v)$ slides the blue component $B$ into $f$. 
(b) $F_t = (w, v)$ mirrors $B$. 
(c) The square vertices form a separation pair. 
We show the separation flip after which $v$ becomes incident to $f$. 
\vspace{-0.4cm}}
\label{fig:flips}
\end{figure}

\mypar{Projected vertices and  conflicts.}
Let $f$ be a face in $\E(G)$ and $v$ be a vertex incident to~$f$. 
We denote by $\pi_f(v)$ the \emph{projected articulation vertex};\label{def:piv}
the articulation point that isolates the maximal (by inclusion) subgraph containing $v$ from the sink of $f$ (if any such exists).
Let $c_u, c_v$ be two corners incident to a face $f$ where $c_u$ is not the top corner. 
We say that $c_u \tor c_v$ is \emph{conflicted} (Figure~\ref{fig:insertion}(a)) in $f$ if: 
the subwalk $\pi_b$ incident to $f$ from $c_v$ through $c_u$ (that does not include the top corner of $f$ or source of $G$) is:
directed towards $c_u$, or intersects a sink corner before it intersects a source corner.
We say $u \tor v$ is conflicted in $f$ whenever for all corners $c_u$ incident to $u$ and $c_v$ incident to $v$, $c_u \tor c_v$ is conflicted.

\begin{lemma}
\label{lemma:conflict}
Let $c_u$ and $c_v$ be corners sharing a face $f$.
If $c_u$ is a top corner, or $c_u \tor c_v$ is conflicted in $f$, the edge $c_u \tor c_v$ violates upward planarity.
Otherwise, we can insert $c_u \tor c_v$ and create corners $c_u^{b}, c_v^{b}$ 
neighbouring $c_u$ and $c_v$ 
incident to $\pi_b$ at an acute angle.
\end{lemma}

\begin{proof}
If $c_u$ is the top corner of $f$, then $c_u \tor c_v$ cannot be drawn $y$-monotone increasing. 
If $c_u \tor c_v$ is conflicted because the subwalk $\pi_b$ is monotone increasing then $c_u \tor c_v$ creates a cycle.
If $c_u \tor c_v$ is conflicted due to some sink corner $c_s$, then the subwalk from $c_v$ to $c_s$ must be monotonely increasing.
Thus, for all drawings $D(G)$ of $\E(G)$, $c_s$ must be higher than $c_v$ (and thus when inserting $c_u \tor c_v$, higher than $c_u$). 
It follows that the edge $c_u \tor c_v$ cannot be drawn $y$-monotone increasing. 

If $c_u \tor c_v$ is not conflicted, the path $\pi_b$ is either monotone decreasing or includes a source corner $c_s'$ before a sink corner. 
Let $c_s'$ be the first such source corner.
In any drawing $D(G)$ of $\E(G)$, the corner $c_v$ is higher than $c_s'$: all other sink corners on $\pi_b$ slightly lower than the corner preceding $c_s'$. So, we can insert $c_u \tor c_v$ at acute angles near $\pi_b$.
\end{proof}

\newpage
\mypar{Update operations.}
We say that an update to $\E(G)$ violates upward planarity if 
it cannot be accommodated such that $\E(G)$ remains upward planar.
Two corners $c_1, c_2$ in $\E(G)$ are \emph{uplinkable} in $\E(G)$ whenever they share a face in $\E(G)$ and 
inserting the edge from $c_1$ to $c_2$ does not violate the upward planarity. 
We dynamically maintain $\E(G)$ subject to the following \emph{combinatorial} and \emph{embedding updates} (see Figure~\ref{fig:operations}) as long as they do not violate upward planarity. 
Combinatorial updates 
modify $G$ (and thus $\E(G)$):
\begin{itemize}[noitemsep, nolistsep]
    \item \textbf{\boldmath Insert$(c_u, c_v)$} for two corners $c_u$ and $c_v$, inserts the edge $c_u \tor c_v$ into $\E(G)$. 
    \item \textbf{\boldmath Delete$(e)$} for an edge $e \in E$, removes it from $\E(G)$.
    \item \textbf{\boldmath Cut$(c_1, c_2)$} for two corners $c_1$ and $c_2$ incident to a vertex $v$, replaces $v$ by $v_1$ and $v_2$ 
    Each of these two vertices becomes incident to a unique consecutive interval of edges incident to $v$ that is bounded by $c_1$ and $c_2$.
    \item \textbf{\boldmath Contract$(e)$} contracts an edge $e \in E$, merging the two endpoints.
    \item \textbf{\boldmath Mirror$(v)$} Mirrors the embedding of the subgraph containing $v$.
\end{itemize}

\noindent When $G$ has one source per connected component, we also support 
embedding updates that only change $\E(G)$ through flips, twists, and slides:
\begin{itemize}[noitemsep, nolistsep]
    \item \textbf{\boldmath Articulation-slide$(w,f,x)$} 
    executes the articulation slide $F_s = (f, w, v)$. 
    \item \textbf{\boldmath Articulation-twist$(w, x)$} executes the articulation twist $F_t = (w, v)$. 
    \item \textbf{\boldmath Separation-flip$(f, g, x, y, v)$} 
    executes the separation flip $F = (f, g, x, y, v)$. 
\end{itemize}
\noindent
Note that, by combining these operations, 
we may perform 
the changes to the embedding known as 
\emph{(Whitney) flips}. These
can transform between any two planar embeddings. 
Similarly, 
the operations above allow to transform between any two upward embeddings of a single-source digraph with the same leftmost edge around the source~\cite[Lemma 6]{bruckner2019spqr}.
\begin{figure}[tb]
\centering
\includegraphics[scale=0.9]{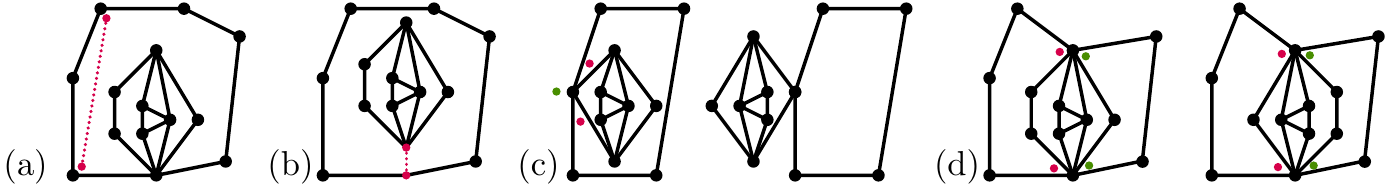}
\caption{
Operations: (a) insertion, (b) cut, (c) articulation-slide, and (d) separation-flip.
}\vspace{-0.4cm}
\label{fig:operations}
\end{figure}

\mypar{Queries.}
Our data structure supports the following queries:
\begin{itemize}[noitemsep, nolistsep]
    \item \textbf{\boldmath UpLinkable$(u, v)$} identifies in $O(\log^2 n)$ time all faces in $\E(G)$ across which $u \tor v$ is uplinkable. It can report these $k$ faces in $O(k)$ additional time.  
    \item \textbf{\boldmath Slide-UpLinkable$(u, v)$} identifies all slides $F_s = (f, w, v)$ and $F_s' = (f', w, u)$ where after the articulation slide, $u \tor v$ is uplinkable in $f$ (or $f'$). 
    \item \textbf{\boldmath Twist-UpLinkable$(u, v)$} reports the at most $2$ articulation twists $F_t = (w, v)$ and $F_t'  = (w, u)$ where afterwards, $u \tor v$ are uplinkable. 
    \item \textbf{\boldmath Separation-UpLinkable$(u, v)$} reports for $u \tor v$ not uplinkable, at least one separation flip $F = (f, g, x, y, v)$ or $F' = (f', g', x', y', u)$ after which $u$ and $v$ are uplinkable.
    \item \textbf{\boldmath One-Flip-UpLinkable$(u, v)$} reports whether there exists 
    two faces $f$ and $g$ and any sequence 
    of flips with only $f$ and $g$ as face arguments, after which $u \tor v$ are uplinkable (if this is the case a constant number of flips will suffice). 
    Note that we do not require that $u$ and $v$ don't share a face, contrary to its undirected analog~\cite{holm2017dynamic}.
\end{itemize}
\newpage
\section{Dynamic Planar Upward Embeddings}
\label{sec:datastructure}
Let $G$ be a single-source upward planar digraph and $\E(G)$ be some combinatorial (upward planar) embedding of $G$. The foundation of our data structure is the data structure by Holm and Rotenberg~\cite{holm2017dynamic} which can maintain a planar $\E(G)$ in $O(\log^2 n)$ subject to all updates. 
In addition, we maintain our data structure in $O(\log^2 n)$ time per update which can reject updates that would violate the upward planarity of $\E(G)$. 

Let $G$ be a digraph and $\E(G)$ be an upward planar combinatorial embedding of $G$. 
We dynamically maintain the face-sink graph $\F(\E(G))$ of $G$ in $O(\log^2 n)$ time per update through the following data structure where we store for every:
\begin{enumerate}[(a), noitemsep, nolistsep]
    \item face $f$, a balanced binary tree $T_f$ on the corners incident to $f$, ordered around $f$.
    \item face $f$, a balanced binary tree $T_f^*$  on the sink corners incident to $f$. $T^*_f$ stores the root-to-leaf paths to the unique top corner in $T^*_f$ (if it exists). \newline In addition, we store a tree $T_f'$ of source corners incident to $f$.
    \item vertex $v$, the tree $T_v$ of corners incident to $v$, plus the tree $T^*_v$ of sink corners incident to $v$, plus a Boolean indicating whether $v$ is critical.
    \item tree $\T$ in $\{ \T^*, \T_1, \ldots T_m \}$ of $\F(\E(G))$ a balanced tree over $\T$ (e.g. a top tree Apx.~\ref{ap:toptrees}).
\end{enumerate}
\noindent
Each corner $c$ in $\E(G)$ maintains a pointer to their location in the above data structures. If $c$ is a sink corner, it maintains a Boolean indicating whether its angle is reflex.

\begin{theorem}
\label{thm:dynamicdatastructure}
We can dynamically maintain $\E(G)$ and $\F(\E(G))$ subject to our updates in $O(\log^2 n)$ time per update 
(rejecting updates that violate the upward planarity of $\E(G)$).
\end{theorem}

\begin{proof}[Proof Sketch]
Let the update change our graph $G$ into some graph $G'$.
By Holm and Rotenberg~\cite{holm2017dynamic}, we can maintain $\E(G)$ in $O(\log^2 n)$ time as long as $\E(G)$ remains planar. 
In addition, we show that w can maintain $\F(\E(G))$ and our data structure in $O(\log^2 n)$ time.
Our data structure then verifies whether $\E(G')$ is upward planar by testing whether the conditions for Theorem~\ref{thm:facesink} are met: and rejects the update accordingly (undoing all changes to $\E(G)$ in $O(\log^2 n)$). We show how to handle edge insertions (Figure~\ref{fig:insertion}(b)):

\mypar{Insert$(c_u, c_v)$}
We consider an edge insertion from $c_u$ to $c_v$ (between vertices $u$ and $v$).
Edge deletion is its inverse and is handled analogously. Given $(c_u, c_v)$, we identify in $O(\log n)$ time the face $f$ incident to $c_u$  as follows:  traverse a pointer from $c_u$ to the unique tree $T_f$ that contains $c_u$ as a leaf. Then we  traverse to the root of $T_f$. We do the same for $c_v$ and detect whether $c_u$ and $c_v$ share a face $f$ in $O(\log n)$ time.

We apply Lemma~\ref{lemma:conflict}: we first test if $c_u$ is the top corner of $f$.
If not, denote by $c^*$ the top corner.
We find the subwalk $\pi_b$ incident to $f$ from $c_u$ to $c_v$ which does not include $c^*$ (or the sink of $G$) in $O(\log n)$ time.
We use $T_f^*$ and $T_f'$ to test if the path $\pi_b$ includes a source corner before a sink corner.
If $\pi_b$ contains neither then it must be a directed path and we check if it is directed from $c_v$ to $c_u$ or vice versa in $O(1)$ additional time.
If it contains at least one sink corner, we test if $c_u \tor c_v$ is conflicted by comparing it to the first source corner on the path.
If $c_u \tor c_v$ is not conflicted, the edge splits $c_u$ into $c_b, c_t$, where $c_b$ is at an acute angle. 
We update our data structure:

\textbf{(a)} The edge splits $f$ into two faces $(f_b, f_t)$ (each a new node in $\F(\E(G'))$).
Let $f_b$ be incident to $\pi_b$ (and thus contain $c_b$).
We construct $T_{f_b}$: $f_b$ is incident to a subsequence $C(1)$ of corners incident to $f$. We obtain $C(1)$ from $T_f$ in $O(\log n)$ time as $O(\log n)$ balanced subtrees. We can merge these balanced binary trees into $T_{f_b}$ in $O(\log^2 n)$ total time.

\textbf{(b)} We construct $T_{f_b}^*$ (and $T_{f_t}^*$) with an identical procedure.
What remains, is to identify the top corners incident to $f_b$ and $f_t$. 
If $c_v$ was the top corner of $f$ then $c_v$ was acute. 
Per definition of acute corners, both $c_b$ and $c_t$ must be two acute corners incident to $f_b$ and $f_t$ respectively. These become new top corners of $f_b$ and $f_t$. 
Otherwise, by Lemma~\ref{lemma:conflict}, $c_b$ is an acute corner and a top corner of $f_b$. 
We test in $O(\log n)$ time whether the top corner of $f$ (if any exists) is incident to $f_b$. If it is, $f_b$ is incident to two top corners and upward planarity must be violated. 
Otherwise, we update the root-to-leaf paths in $T_{f_b}^*$ and $T_{f_t}^*$ accordingly. 
We construct $T_{f_b}'$ and $T_{f_t}'$ analogously.

\textbf{(c)} We update $T_u$ and $T_v$ in an analogous manner.
Finally, we update the Booleans of $u$ and $v$: 
it may be that because of the insertion the vertex $u$ became an internal vertex in $G$.  
Moreover, it may be that $v$ is a vertex internal in $G$ and that after the insertion, $v$ became incident to a sink corner (and hence critical). We obtain $u$ (and $v$) in $O(\log n)$ time, test whether the vertices are internal in $O(1)$ additional time, test whether $T^*_u$ and $T^*_u$ are empty in $O(1)$ time, and adjust the Booleans accordingly.

\textbf{(d)} Finally, we update $\F(\E(G))$.
The faces $f_b$ and $f_t$ each are a new node in $\F(\E(G))$. 
The parent nodes of these faces correspond to top corners in $T_{f_b}^*$  and $T_{f_t}^*$. Given a top corner $c^*$, we identify the vertex $w$ incident to $c^*$ in $O(\log n)$ time.
By (b), $w$ is incident to $f_t$ and we set $w$ to be the parent of $f_t$. We update the corresponding tree in $\F(\E(G))$ in $O(\log n)$ time (Property~\ref{prop:toptree}). 
In the graph $\F(\E(G))$, the children of $f_b$ and $f_t$ can be obtained by separating the children of $f$ around two corners in any ordered embedding of $\F(\E(G))$, which can be done in $O(\log n)$ time (Property~\ref{prop:toptree}). 
Since we only inserted child-to-parent edges,  $\F(\E(G))$ is still a forest.
We test whether each tree $\T$ of $\F(\E(G))$ has a unique critical vertex (or the outer face) as its root. Specifically, this operation affects at most two trees $\T_a, \T_b$ of $\F(\E(G))$ and at most two vertices.
We check if these constantly many objects are still valid in $O(\log n)$ time and if not, then we reject the update.

\begin{figure}[h]
\centering
\includegraphics{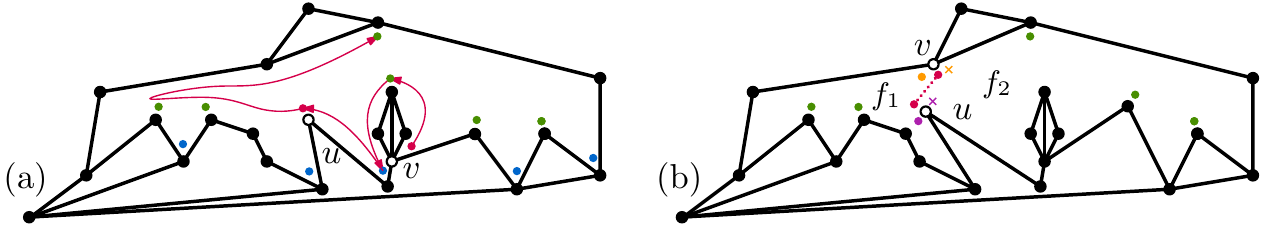}
\caption{(a) Insertion between two corners (red).
The path from $c_2 \tor c_1$ excluding the top corner first goes through a sink corner (green) and then through a source corner (blue). This implies that $c_1 \tor c_2$ is conflicted.
(b) By Lemma~\ref{lemma:conflict}, if $u \tor v$ are not conflicted we insert the edge such that the corner incident to $f_1$ is acute.
}
\label{fig:insertion}
\end{figure}

\newpage

\noindent
Finally we show how to do (along similar lines): 
Delete$(e)$ / Cut$(c_1, c_2)$ / Contract$(e)$ /  
Articulation-Slide$(f, w, v)$ / Articulation- Twist$(w, v)$ / Separation-flip$(f, g, x, y, v)$ / Mirror$(v)$.   
For each of these operations, we show how to maintain properties (a), (b), (c) and (d) of our data structure. We show that maintaining these properties allows us to verify whether the update violates upward planarity.

\mypar{Delete$(e)$}
Observe that an edge-deletion can only violate our conditions whenever the deleted edge $u \tor v$ creates a new source $v$.
We can test this in $O(\log n)$ time using $T_v$ and test if the deletion splits $G$ into two connected components. If we find that one connected component has two sources, we reject the update. Otherwise, all updates to our data structure are simply the inverse of the updates for edge insertion and handled analogously.

\mypar{Cut$(c_1, c_2)$ / Contract$(e)$.}
Let $(c_1, c_2)$ be incident to some vertex $v$. 
The Cut operation creates two new vertices $v_1$ and $v_2$ from a vertex $v$ and joins faces $f_1$ and $f_2$ into some new face $f$.
In the special case where $f_1 = f_2$, the cut operation must introduce a new source:
we test if $v_1$ and $v_2$ are connected: if so the update is rejected. Otherwise, we gained a new connected component with a unique source.
we show how to update our data structure for the cut operation. Contraction is its inverse and is handled analogously:

\textbf{(a)} + \textbf{(b)}
The cut operation deletes four corners which we can delete from our data structure in $O(\log n)$ time (Property~\ref{prop:toptree}). 
We then create the trees $T_f$, $T_f^*$ and $T_f'$ in $O(\log n)$ time by merging the corresponding trees of $f_1$ and $f_2$ (Property~\ref{prop:toptree}). 
Finally, the cut operation introduces two new corners. 
We test in $O(1)$ time whether these new corners are top, spike or source and we insert them in the corresponding trees.

\textbf{(c)} The trees $T_{v_1}, T_{v_1}^*, T_{v_2}, T_{v_2}^*$ are analogously created in  $O(\log^2 n)$ time by splitting the trees $T_v$ and $T_v^*$.
 Then, we test whether $v_1$ and $v_2$ are internal vertices of $G$ in $O(\log n)$ additional time by traversing the cyclical ordering on the edges outgoing of $v$. 

\textbf{(d)} The faces $f_1$ and $f_2$ become one node $f$ in $\F(\E(G))$, which inherits all their children. 
We test in $O(\log n)$ time, using the updated $T_f^*$, if the face $f$ is incident to a unique top corner. 
If not, then either $f_1$ or $f_2$ was the outer face, or the face-sink graph must be invalid and the update is rejected. 
Otherwise, we have identified the unique tree $\T$ in $\F(\E(G))$ that contains $f$.
We test in $O(\log n)$ time whether the new vertex $v_2$ has a sink corner incident to $f$ (making it a child of $f$). 
Finally, we update the balanced binary tree over the tree $\T$ containing $f$ in $O(\log n)$ time (Property~\ref{prop:toptree}). We show Mirror$(v)$ at the end.

\mypar{Articulation-Slide$(f, w, v)$ / Articulation-Twist$(w, v)$.}
Note that all combinatorial changes incurred by an articulation slide or twist, also occur in a separation flip.
Specifically, an articulation slide may alter at most four corners in $\E(G)$, and causes a subwalk around a face $g$ to become a subwalk around a face $f$.
An articulation twist selects mirrors a component.
A separation flip performs all these three combinatorial changes, twice.
Our procedure for the separation flip will specify all the changes to our data structure for these occurrences, and thus for the aritculation slide and twist also.

\newpage
\mypar{Separation-flip$(f, g, x, y, v)$.}
Let the face $g$ share corners $c_g^x$ and $c_g^y$ with $x$ and $y$, respectively.
Similarly, let $f$ share corners $c_f^x$ and $c_g^y$.
We illustrate the operation by Figure~\ref{fig:separationflip}. We update our data structure as follows:

\textbf{(a)}
The corners $c_g^x$ and $c_g^y$ bound a contiguous sequence $C(1)$ of corners which:
before the update are incident to $g$ and afterwards incident to $f$. 
We obtain this subsequence as at most $O(\log n)$ binary subtrees of $T_{g}$, and merge them into $T_{f}$ in $O(\log^2 n)$ total time.
We do the same for the corners between $c_f^x$ and $c_f^y$. 

\textbf{(b)}
We update $T_{g}^*, T_g'$ and $T_{f}^*, T_f'$ in a similar fashion where $C^*(1)$ is the sequence of sink corners in between $c_g^x$ and $c_g^y$. 
It remains to ensure that both trees have a marked root-to-leaf path to the unique top corner $T_g^*$ and $T_f^*$.
If $f$ and $g$ switched top corners, then they must have laid on $C^*(1)$ (and the corresponding $C^*(2)$). 
We test in $O(\log n)$ time and adjust the root-to-leaf paths in both $T_{g}^*$ and $T_{f}^*$ accordingly. 

\textbf{(c)} For the tree $T_x$ the separation flip selects and inverts a contiguous subsequence of corners incident to the vertex $x$.
We support this operation through the following trick: 
we maintain for each $T_x$ both its clockwise rotation of corners, and its counterclockwise rotation in two trees $T_x$ and $T_x'$. 
Denote by $C_x$ the subsequence of corners in between $c_g^x$ and $c_f^x$: 
this subsequence consists of at most $O(\log n)$ subtrees in $T_x$ and $T_x'$.
To invert this order, we simply interchange these subtrees and rebalance both trees in $O(\log^2 n)$ total time.
We do the same for the tree $T_y$ (with $T_y'$). 
For all other trees in the flip component $C_2$ separated by $(x, y)$: we still have a tree for their clockwise and counterclockwise rotation (although they have interchanged) and we hence do not have to update them. 

\textbf{(d)}
We update $\F(\E(G))$ as follows:
Let $f$ be contained in $\T_f$ and $g$ be contained in $\T_g$.
All spike corners $C^*(1)$ receive $f$ as their parent instead of $g$.
This operation is \emph{directly} supported by our choice of balanced tree over $\T_f$ and $\T_g$ (Property~\ref{prop:toptree}) in $O(\log n)$ time.
The separation flip destroys the corners $c_g^x, c_g^y, c_f^x, c_f^y$ and replaces them with four new corners.
We compute these in $O(1)$ time and check whether they are sink corners.
If so, then we insert the corresponding relation into $\F(\E(G))$ in $O(\log n)$ time per new corner. 
Finally, if $f$ and $g$ swapped parents in $\F(\E(G))$, we execute this swap in $O(\log n)$ time.

\mypar{Mirror$(v)$}: The mirror operation is described by bullet \textbf{(c)} of the separation flip. 
\end{proof}

\begin{figure}[h]
\centering
\includegraphics{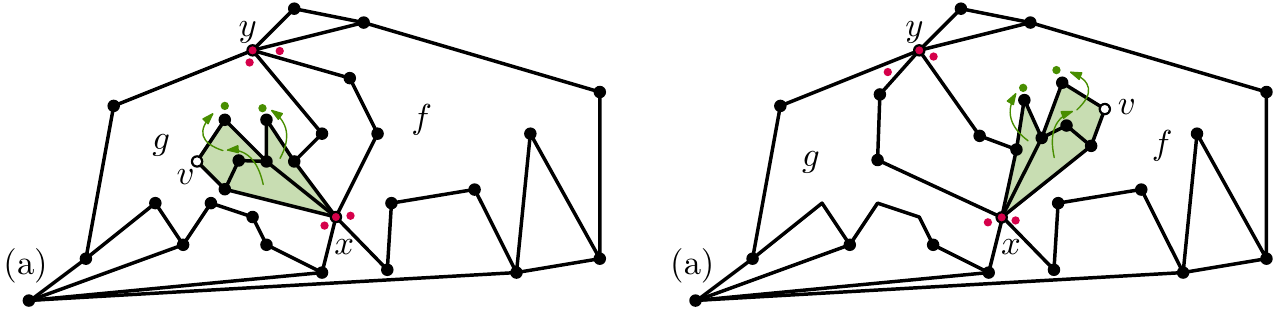}
\caption{
A separation flip $F = (f, g, x, y, v)$.
Note that the green subtree under $g$ in the face-sink graph becomes a subtree under $f$ instead.
}
\label{fig:separationflip}
\end{figure}

\newpage
\section{Supporting uplinkability queries}
\label{sec:queries}
We show how we support uplinkability queries.
Specifically, for each of our queries we spend  $O(\log^2 n)$ such that afterwards, we can report in $O(k)$ time the first $k$ items of the output.
The exception is the Separation-UpLinkable$(u, v)$ query where in $O(\log^2 n)$ time we return one separation flip.
We maintain the same data structure as in Section~\ref{sec:datastructure} which includes the structure Holm and Rotenberg~\cite{holm2017dynamic} which maintained a planar (undirected) graph $G'$ and its combinatorial embedding $\E(G')$ subject to linkablility queries. 

Let $G'$ be any planar graph and $\E(G')$ be a planar embedding of $G'$.
Holm and Rotenberg~\cite{holm2017dynamic} dynamically maintain $\E(G')$ supporting the following queries in $O(\log^2 n)$ time:\footnote{
The latter two queries were encompassed in their definition of \emph{one-flip-linkable})}

\begin{itemize}[noitemsep, nolistsep]
    \item \textbf{\boldmath Linkable$(u, v)$} returns for two vertices $u,  v$ all corner pairs across which $u$ and $v$ are linkable in a data structure. This data structure can, for any $i$, return the $i$'th face in $S$ in $O(\log n)$ time and report all faces up to the $i$'th one in $O(i)$ time.
    \item \textbf{\boldmath Slide-Linkable$(u, v)$} returns 
    all articulation slides $F_s = (f, \pi(v), v)$ where after $F_s$, $u$ and $v$ share a face $f$.
    Their output can, for any $i$, return the $i$'th slide (around $\pi(v)$) in $O(\log n)$ time and report $i$  consecutive articulation slides in $O(i)$ time.
    \item \textbf{\boldmath Separation-Linkable$(u, v)$} returns for two vertices $u, v$ for which Linkable$(u, v)$ and Slide-Linkable$(u, v)$ are empty, a separation flip after which $u$ and $v$ share a face.
\end{itemize}

\mypar{The key difference between linkability and uplinkability.}
In a planar embedding $\E(G)$, two vertices are linkable across any face they share.
In the upward planar setting this is not true (Figure~\ref{fig:uplinkable} in Apx.~\ref{ap:queries}) which significantly complicates uplinkability queries.
To answer UpLinkable$(u, v)$ and Slide-UpLinkable$(u, v)$ 
we need to identify the subset of faces across which $u \tor v$ are uplinkable and become uplinkable after a slide, respectively. 
For Separation-UpLinkable$(u, v)$, we find from all possible separation flips, at least one making $u$ and $v$ uplinkable.
Before we start our analysis we observe the following result from the definition of conflicting vertices (any path $\pi_b$ must include the mentioned $\pi_b^*$):
\begin{observation}
\label{obs:shortestpath}
Let $u$ and $v$ share a face $f$ and $u \tor v$ be not conflicted. 
Denote by $\pi_b^*$ the shortest subwalk incident to $f$ from $v$ to $u$ which excludes the top corner of $f$ and the source of $G$. 
Then $\pi_b$ cannot reach a sink corner before reaching a source corner or $u$. 
\end{observation}

\mypar{UpLinkable$(u, v)$} reports all faces across which $u \tor v$ doesn't violate upward planarity.

\begin{lemma}
\label{lem:uplinkable}
Let $G$ be a single-source digraph, $\E(G)$ be upward planar, and let $f$ be a face incident to $u$ and $v$.
If $v \not < u$,  
the edge $u \tor v$ may be inserted across $f$ if and only if: 
\begin{enumerate}[noitemsep, nolistsep]
    \item[(i)] $u$ is not the vertex corresponding to the top corner of $f$, 
    \item[(ii)] $u\!\rightarrow\! v$ is not conflicted in $f$, and
    \item[(iii)]
$f$ is incident to an edge directed towards $v$.
\end{enumerate} 
If conditions $(i)$--$(iii)$ hold, $v \not< u$ if and only if:
\begin{enumerate}[noitemsep, nolistsep]
    \item[(iv)] there is no directed path on the boundary of $f$ from $v$ to $u$.
\end{enumerate} 
\end{lemma}

\begin{proof}

Assume $v \not < u$. 
If $u$ lies on the top corner of $f$, or if $u \tor v$ is conflicted in $f$ then Lemma~\ref{lemma:conflict} implies that inserting $u \tor v$ across $f$ violates upward planarity. 
Otherwise, denote by $c_u$ and $c_v$ two corners (incident to $u$ and $v$ respectively) where $c_u \tor c_v$ is not conflicted. 
Denote by $\pi_b$ the subwalk incident to $f$ from $c_u$ and $c_v$ excluding the top corner and the unique source.
Denote by $\pi_t$ the complement of $\pi_b$. 
 The edge $c_2 \tor c_1$ splits $f$ into two faces $f_t$ and $f_b$ (incident to $\pi_t$ and $\pi_b$ respectively).\footnote{We associate $t$ and $b$ with ``top'' and ``bottom''; in the outer face, this sense of direction is reversed.}
We make a case distinction based on the edges incident to $c_v$. Per case, we show whether $c_u \tor c_v$ may be inserted:

\textbf{\boldmath Case 1: $c_v$ is incident to two outgoing edges - NO.} 
See Figure~\ref{fig:uplinkableproof}(a + b).
Consider the face $f$ in the face-sink graph and the tree $\T$ that contains $f$.
After the insertion, $f_b$ cannot be incident to a top corner: $f_b$ is per construction not incident the top corner of $f$ and the insertion creates no additional sink corners. 
Thus, $f_b$ becomes a root of a tree in the face-sink graph. However, by Theorem~\ref{thm:facesink} the only roots in the face-sink graph may be critical vertices or the outer face. Thus, inserting $c_u \tor c_v$ invalidates the face-sink graph.

\textbf{\boldmath Case 2: $c_v$ is incident to one outgoing edge and one incoming edge - YES.}
The outgoing edge must be part of $\pi_t$. 
We claim that the face-sink graph remains valid after inserting $c_u \tor c_v$
(Figure~\ref{fig:uplinkableproof}(c)):
the corner $c_v$ is split into two corners $c_t$ and $c_b$ where $c_b$ is an acute sink corner: thus a top corner. 
After the insertion, all children of $f$ become children of either $f_t$ or $f_b$, with possibly the exception of $u$: when $u$ was a spike, it now became a critical vertex.
The face $f_t$ is either the outer face, or incident to the top corner of $f$ and thus any subtree rooted at $f_t$ is part of a valid tree in $\mathcal{F}(\E(G))$. 
The face $f_b$ is incident to a unique top corner which is incident to $v$. The vertex $v$ must be  a critical vertex of $G$  and the subtree of $f_b$ is thus also part of a valid tree in $\mathcal{F}(\E(G))$.

\textbf{\boldmath Case 3: $c$ is incident to two incoming edges - YES.}
Assume that $c$ is the top corner of $f$. 
By inserting $u \tor v$, we split $c$ into two top corners incident to both $f_u$ and $f_l$. 
Thus, in the face-sink graph, $f$ gets replaced by $f_u$ and $f_v$ who partition its children and $\mathcal{F}(\E(G))$ is thus still valid. 
Assume otherwise that $c$ is a spike in $f$. 
Because $u \tor v$ is not conflicted, the insertion splits $c$ into a top corner (incident to $f_b$) and spike corner (incident to $f_t$). 
Thus, in the face sink-graph, $f$ is simply replaced by the path $f_b \tor v \tor f_t$ (where $f_b$ and $f_t$ partition the children of $f$). Thus, $\mathcal{F}(\E(G))$ remains valid. 

\textbf{\boldmath Case 4: $c$ is incident to exactly one outgoing edge - NO.}
Here, $v$ is the unique source of $G$. 
Inserting $u \tor v$ introduces a directed cycle in $G$, violating upward planarity. 

\textbf{\boldmath Case 5: $c$ is incident to exactly one incoming edge - YES.}
This case is identical to case where we create a top corner incident to $f_b$. Thus, in $\mathcal{F}(\E(G))$, $f$ is replaced by the path $f_b \tor v \tor f_t$ (and $f_b$ and $f_t$ partition the children of $f$). 

There exists a corner pair $(c_u, c_v)$ to which Case $2$, $3$, or $5$ applies, if and only if there exists at least one directed towards $v$ incident to $f$ which concludes the first part.

For the second part of the proof, 
assume conditions $(i)$--$(iii)$ and consider the path $\pi_b^*$.
Conditions $(ii)$ and $(iii)$ together with  Observation~\ref{obs:shortestpath}
imply that $\pi_b^*$ starts with an edge directed towards $v$ (and thus it cannot be a path from $v$ to $u$). 
If $v < u$ then there must exist a directed path $\pi^*$ from $u$ to $v$. The path $\pi^*$ together with $\pi_b^*$ encloses some face $g$.
Assume for the sake of contradiction that $g \neq f$ (i.e. $\pi^*$ is not on the boundary of $f$) then $f$ must be contained in $g$.
If $f$ is the outer face, it cannot be contained in any other face: contradiction.
If $f$ is not the outer face, then the top corner of $f$ is contained within $g$. 
However in any upward drawing of $G$, the top corner must be higher than $u$, this implies that the directed path $\pi^*$ cannot be drawn $y$-monotone increasing: contradiction.
%
\end{proof}
\begin{figure}[bt]
\centering
\includegraphics[scale=1]{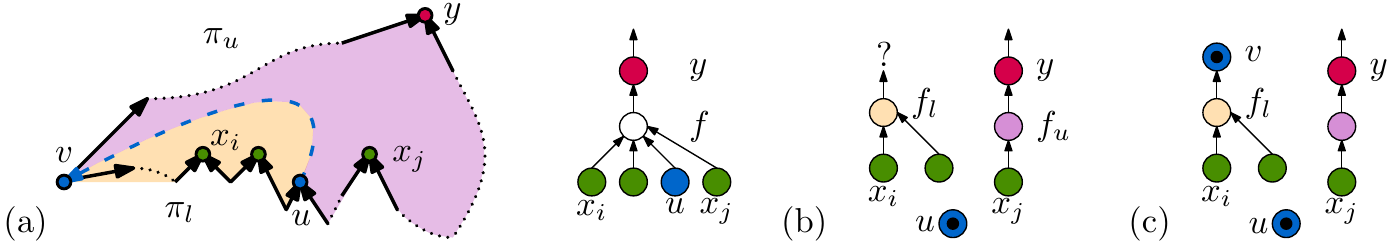}
\caption{
(a) $f$ is split into $f_l$ and $f_u$ (orange and purple). We show the original tree $T$ in $\mathcal{F}(\E(G))$.
(b) The children of $f$ are partitioned by $f_l$ and $f_u$, except for $u$ which becomes 
critical. The  face $f_l$ becomes parent-less. (c) $f_l$ receives  the critical vertex $v$ as its parent.
}\vspace{-0.4cm}
\label{fig:uplinkableproof}
\end{figure}

\noindent
We briefly show how Lemma~\ref{lem:uplinkable} enables uplinkability queries:

\begin{lemma}\label{lemma:uplinkable_query}
We can support the UpLinkable$(u, v)$ query in $O(\log^2 n + k)$ time.
\end{lemma}

\begin{proof}
First, assume that $v \not < u$ in $G$ (we lift this assumption at the end). By Linkable$(u, v)$, we obtain in $O(\log^2 n)$ time all faces incident to both $u$ and $v$ as the data structure $S$. 
For all faces $f \in S$, by Lemma~\ref{lem:uplinkable}, $u \tor v$ may be inserted into $\E(G)$ if and only if properties $(i)$--$(iii)$ hold.
We can test for $f$ these properties in $O(\log n)$ time using our data structure:
\begin{enumerate}[noitemsep, nolistsep]
    \item[$(i)$]   using $T_f^*$ we identify the top corner of $f$ in $O(\log n)$ time,
    \item[$(ii)$] by binary search along $u$ and $v$ in $T_u$ and $T_v$ we find the path $\pi_b^*$ in $O(\log n)$ time;
    using $T_f^*$ and $T_f'$ we then test if  a sink precedes a source on $\pi_b^*$.
    \item[$(iii)$] using $T_v'$ we find an edge in $f$ incident towards $v$ in $O(\log n)$ time.
\end{enumerate}

If $S$ contains more than one element, we note that the set $S_1 \subseteq S$ of faces $f$ which have property $(i)$ must be a contiguous subset of $S$: all faces $f'$ where $u$ is the top corner of $f'$ must have two edges directed towards $u$ incident to $f'$.
Similarly, the set $S_3 \subset S$ of faces $f$ which have property $(iii)$ must be a contiguous subset of $S$.
Finally, if $S$ contains more than one element and $f \in S$ has an edge directed towards $v$, then $u \tor v$ can only be conflicted in $f$ if $v < u$ which we assumed to be not true. 
It follows that if $S$ contains more than one element then the set $S^*$ of faces which have properties $(i)$, $(ii)$ and $(iii)$ must be a contiguous subset of $S$ which we can identify in $O(\log n)$ additional time. 
By Holm and Rotenberg~\cite{holm2017dynamic}, we can report the first $k$ elements in this output in $O(k)$ additional time. 

What remains is to lift the assumption that $v \not < u$. 
Suppose that the above output $S^*$ is not empty.
Then we simply select one of the faces $f$ of the output. 
by Lemma~\ref{lem:uplinkable}, it is enough to 
test condition $(iv)$, that is, whether there is a directed path from $v$ to $u$ on the boundary of $f$.
This can be done in $O(\log n)$ time by testing whether there are any sink  or source corners  between them. We report either all or none of $S^*$ accordingly.
\end{proof}

\noindent
Before we summply the proof for the remaining queries, we wish to highlight the main approach and challenges for answering these queries (Figure~\ref{fig:finalfigure} (a+ b + c):

\mypar{Slide-UpLinkable$(u, v)$:} by Slide-Linkable$(u, v)$, we find (an implicit representation of) all slides $F_s = (f, w, v)$ around $w$ (where $u$ is not in the same component as $y$ in $G \backslash \{ w \}$) after which $u$ and $v$ share a face, as an ordered set by $\Omega_v$  in $O(\log^2 n)$ time.
We first identify the set $\Omega_v' \subseteq \Omega_v$ of slides that do not violate upward planarity: we prove (Lemma~\ref{lemma:legal_arti}) that $\Omega_v'$ is a contiguous subset. 
Then, we want to identify the subset $\Omega^*_v \subseteq \Omega_v'$ of slides which are part of our output. We show that the faces in $\Omega'_v$ where all conditions of Lemma~\ref{lemma:uplinkable_query} are met, must be a contiguous subset and we output $\Omega^*_v$ accordingly. 
Doing a near-identical procedure for slides $F_s' = (f, w', u)$ then concludes our result.

\mypar{Twist-UpLinkable$(u, v)$:} By using a subroutine of Slide-Linkable$(u, v)$, we identify an implicit representation of the set $\Lambda_v$ of twists $F_t = (w, v)$ in $O(\log^2 n)$ time.
We show that $u \tor v$ are either uplinkable after all of $\Lambda_v$ or none of $\Lambda_v$; as a consequence, it is uplinkable across $f$ after any twist if and only if it is uplinkable after a twist in the vertex $\pi_f(v)$ (defined on page~\pageref{def:piv}). We use Lemma~\ref{lem:uplinkable} to verify this for one twist 
through $\pi(v)$, and output it 
accordingly. For twists $F_t' = (w', u)$ we do an identical procedure. 

\mypar{Separation-UpLinkable$(u, v)$:} By using a variant of Separation-Linkable$(u, v)$ we consider all separation flips $F = (f, g, x, y, v)$ where $u$ is incident to the face $f$, and $v$ to $g$. 
Of these flips, we obtain the unique flip $F^* = (f, g, x^*, y^*, v)$ where $x^*$ and $y^*$ are closest to $v$, and the subgraph containing $v$ is minimal with respect to inclusion.
We prove (Lemma~\ref{lemma:only_sep_you_need}) that if performing $F^*$ violates upward planarity, then any such $F$ must violate upward planarity. 
We show the same statement holds for uplinkability of $u \tor v$. We perform $F^{\ast}$, and use Lemma~\ref{lem:uplinkable} to test whether we may now insert $u \tor v$, thus deciding the question.

\mypar{One-Flip-UpLinkable$(u, v)$:} we show that if there exists faces $f$ and $g$, and any sequence of flips, slides, or twists (as above) involving only $f$ and/or $g$ after which $u \tor v$ may be inserted, then this sequence has constant size. 
If any such sequence exists we output it in $O(\log ^2 n)$ time. 
Upward planarity is more complicated than planarity in this regard, as strictly more than one flip may be necessary when the vertices already share a face.


\begin{figure}[h]
\centering
\includegraphics[]{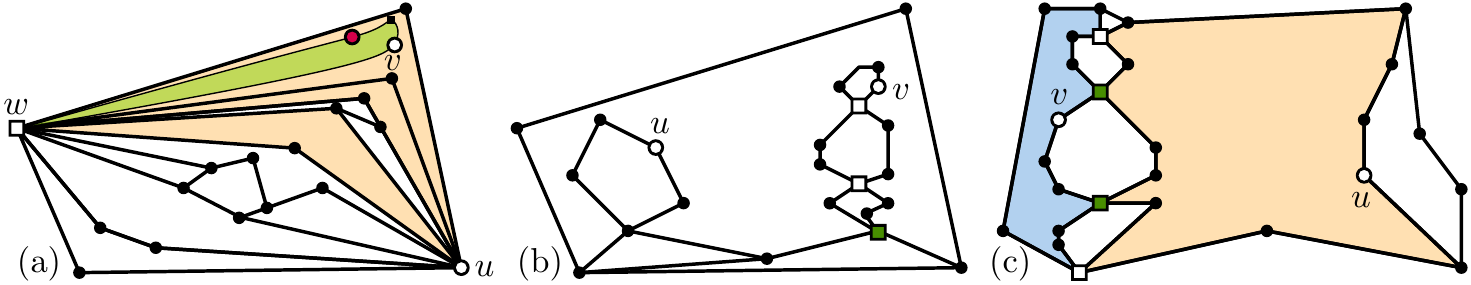}
\caption{
(a) For articulation slides, the subset $\Omega^*_v$ where afterwards, $u \tor v$ may be inserted is shown in yellow.
(b) There exists a twist $F_t = (w, v)$ for every square vertex $w$. 
We only need to consider last such $w$ (green). 
(c) There exists a separation flip $F$ for each pair of squares.
We show only need to consider the green pair $(x^*, y^*)$.
}
\label{fig:finalfigure}
\end{figure}
\subsection{The remaining queries}

Finally, we show how to support the remaining uplinkability queries, one by one:

\begin{figure}[h]
\centering
\includegraphics{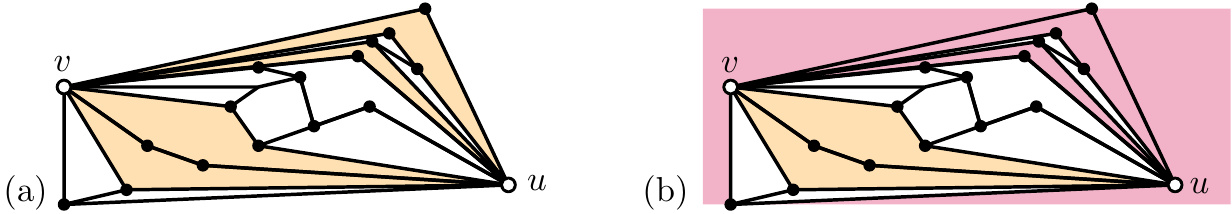}
\caption{(a)  $u$ and $v$ share four faces which are not the outer face of their biconnected component. (b) 
The directed edge $u \tor v$ cannot be  inserted across red shared faces.
}
\label{fig:uplinkable}
\end{figure}

\mypar{Slide-UpLinkable$(u, v)$} reports all articulation slides $F_s' = (f, w', u)$ and $F_s = (f, w, v)$ where after the slide, $u \tor v$ is uplinkable across $f$ in the embedding.
\begin{observation}
\label{obs:unique_arti}
For all $F_s = (f, w, v)$, the vertex $w$ is identical.
\end{observation}

\noindent
Denote by $\Omega_v$ the articulation slides $F_s = (f, w, v)$ (where $u$ is not in the same component as $v$ in $G \backslash \{ w \}$) where afterwards $u$ and $v$ share the face $f$. 
We do the following:
we first show how to obtain $\Omega_v$ in $O(\log^2 n)$ time, ordered around $w$. 
We then identify the subset $\Omega_v'$ of slides $F_s$ that do not violate upward planarity.
Finally, we identify the subset $\Omega^*_v$ of slides $F_s^*$ where after $F_s^*$, $u \tor v$ can be inserted without violating upward planarity.

\begin{lemma}
\label{lemma:find_arti-flips}
We can identify the set $\Omega_v$ in $O(\log^2 n)$ time.
In addition, given any integer~$i$, we can return the $i$'th face in $\Omega_v$ (ordered around $w$) in $O(\log n)$ additional time.
\end{lemma}

\begin{proof}
Holm and Rotenberg~\cite{holm2017dynamic} show how to find for any pair of vertices $v$ and $u$, the articulation vertex $w$ as described in Observation~\ref{obs:unique_arti}.
We test in $O(\log^2 n)$ time whether $u$ and $v$ are both in $G \backslash \{ w \}$ using the cut operation. Given $w$, we can query Linkable$(u, w)$ to obtain all faces that $v$ can slide into such that $u$ and $v$ afterwards share a face.
Per definition of Linkable$(u, w)$, these faces are returned in a data structure $S$ where for any index $i$, we can get the $i$'th face in $S$ (ordered around $w$) in $O(\log n)$ additional time. 
\end{proof}

\begin{lemma}
\label{lemma:legal_arti}
Let $G$ be a single-source digraph and $\E(G)$ be an upward planar.
Let $F_s = (f, w, v)$ be an articulation slide.
Let $C_1$ be the component of  $G \backslash \{ w \}$ containing $v$ and $g$ be the face in $C_1$ incident to $v$. 
Let $C_2$ be the component of  $G \backslash \{ w \}$ containing $f$.
$F_s$ does not violate upward planarity if and only if:
$C_1$ does not contain the source of $G$ and
$f$ and $g$ are incident to an edge in $C_2$ directed out of $w$.
\end{lemma}

\begin{figure}[htb]
\centering
\includegraphics{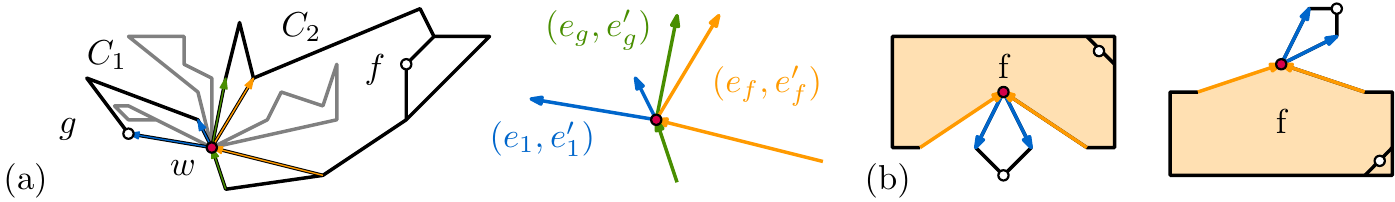}
\caption{
(a) A vertex $w$ and several flip components including $C_1$ and $C_2$ (in black).
We consider only the edges in $C_1$ and $C_2$ incident to either $f$ or $g$.
(b) If $e_f, e_f'$ are directed towards $w$, then either $\E(G)$ violates upward planarity either before or after the slide.
}
\label{fig:legal_arti}
\end{figure}
\begin{proof}
The proof is illustrated in Figure~\ref{fig:legal_arti}.
Denote by $\pi_1$ the outer cycle of $C_1$  and let it contain the unique source of $G$.
If $\E(G)$ is upward planar then the unique source of $G$ is incident to the outer face. Per definition, the outer face cannot be contained in $\pi_1$ and must thus be incident to $\pi_1$ (i.e. $g$ is the outer face).
After sliding $C_1$ into $f$, the unique source is no longer incident to $g$ which violates upward planarity. 

Let $\pi_1$ not contain the unique source. 
We show that $C_1$ may slide into $f$ if and only if $f$ and $g$ are incident to an edge in $C_2$ directed out of $w$.
Denote by $e_1, e_1'$ the edges of $\pi_1$ incident to $w$. 
These must be directed outwards of $w$: otherwise $\pi_1$ contains the unique source of $G$.
Similarly, $\pi_1$ cannot contain the top corner of $g$.
Denote by $e_f, e_f'$ the edges in $C_2$ incident to $f$ and $w$ and by $e_g, e_g'$ the edges in $C_2$ incident to $g$ and $w$. Let $f$ or $g$ not be incident to an edge of $C_2$ directed out of $w$ (Figure~\ref{fig:legal_arti}(b)):

\textbf{\boldmath Let $e_f, e_f'$ be both directed towards $w$}.
If the angle between $e_f, e_f'$ is reflex then before the slide $e_1, e_1'$ are edges directed outwards of $w$ which point below $w$ and thus $\E(G)$ was not upward planar. 
If otherwise their angle is convex then after the slide $e_1, e_1'$ are edges directed towards $w$, pointing below $w$ which violates upward planarity.

\textbf{\boldmath Let $e_g, e_g'$ be both directed towards $w$}.
If the angle between $e_g, e_g'$ is convex then before the slide $e_1, e_1'$ are edges directed outwards of $w$ which point below $w$ and thus $\E(G)$ was not upward planar. 
If otherwise their angle is reflex then after the slide $e_1, e_1'$ are edges directed towards $w$, pointing below $w$  which violates upward planarity.

\textbf{\boldmath Suppose otherwise that  $e_f, e_f'$ and $e_g, e_g'$ both contain at least one edge directed outwards of $w$.}
We show that the slide does not violate the face-sink graph:
 $e_1$ and $e_1'$ are directed outwards of $w$.
 For every other flip-component $C_3$ incident to $w$, their edges incident to $w$ and $f$ must point outwards of $w$ (else $C_3$ contains the source of $G$).
 Thus, after the slide, there are no fewer or extra sink corners incident to $w$ and $f$ (or $w$ and $g$). 
 Because $\pi_1$ does not contain a top corner of $g$, $f$ and $g$ remain incident to the same top corners they were incident to before the flip (if any).
All children of $g$ that lie on $\pi_1$ become children of $f$, and thus remain part of a valid tree in the face-sink graph.
\end{proof}

\begin{lemma}
\label{lemma:slideuplinkable}
We support the Slide-UpLinkable$(u, v)$ query in $O(\log^2 n + k)$ time.
\end{lemma}

\begin{proof}
We show that in $O(\log^2 n)$ time we can identify all slides $F_s = (f, w, v)$
such that after $F_s$, $u$ and $v$ are uplinkable.
Moreover, we show that in $O(k)$ time, we can report the first $k$ such slides.
We do a near-identical procedure for all slides $F_s^u = (f, w, u)$ to prove the lemma. 
Just as in Lemma~\ref{lemma:uplinkable_query} we first assume that $v \not < u$.
If our proposed output is not empty, we then verify whether $v < u$ in a similar fashion.

We find (by Lemma~\ref{lemma:find_arti-flips}) in $O(\log^2 n)$ time the set $\Omega_v$.
All such slides $F_s$ occur around a single articulation vertex $w$ (Observation~\ref{obs:unique_arti}) and we obtain $\Omega_v$ ordered around $w$. 
Next, we identify the subset $\Omega_v'$ of slides that do not violate upward planarity. 
Specifically, let $C_1$ be the flip component that contains $v$ and $C_2$ the flip component that contains $u$.
We first test if $C_1$ contains the unique source of $G$ as follows:
we identify in $O(\log n)$ time the pair of corners incident to $w$ that separate $C_1$ from the remainder of $G$. 
In $O(\log^2 n)$ time, we split $w$ along this corner so that $C_1$ is a separate connected component (Theorem~\ref{thm:dynamicdatastructure}). We then
use the Holm and Rotenberg~\cite{holm2017dynamic} data structure to test if the source and $v$ are still connected in $G$. 
If so, we terminate our query algorithm and conclude that our output contains none of $\Omega_v$.

Otherwise, we revert the vertex split. 
Consider a slide $F_s = (f, w, v) \in \Omega_v$ with $v$ incident to a face $g$ in $C_1$ and $u$ incident to the face $f$ in $C_2$. 
By Lemma~\ref{lemma:legal_arti}, $F_s$ violates upward planarity if and only if $f$ and $g$ are incident to an edge directed outwards of $w$ (in $C_2$).
We test if $g$ is incident to an edge (in $C_2$) directed outwards of $w$ in $O(\log n)$ time by binary searching along the radial ordering of edges incident to $w$ (Theorem~\ref{thm:dynamicdatastructure}). 
If $g$ is not,  we conclude that our output contains none of $\Omega_v$.
If $g$ is incident to such an edge, then the subset $\Omega_v' \subseteq \Omega_v$ of slides $F_s' = (f', w, v)$ where $f'$ is incident to an edge directed outwards of $w$ (in $C_2$) must be a contiguous subsequence of $\Omega_v$. We identify this sequence in $O(\log^2 n)$ additional time using binary search.

Given $\Omega'_v$, we want to find the subset $\Omega^*_v$ of slides involving $v$ that are part of our output. 
For all  slides $F_s \in \Omega'_v$ we know by Lemma~\ref{lem:uplinkable} that $u \tor v$ is uplinkable in $f$ iff:
\begin{enumerate}[$(i)$, noitemsep, nolistsep]
    \item $u$ is not incident to the top corner of $f$, 
    \item $u \tor v$ is  not conflicted in $f$, and 
    \item $f$ is incident to an edge directed towards $v$. 
\end{enumerate}
Now note that:
\begin{enumerate}[$(i)$, noitemsep, nolistsep]
\item Just as for uplinkability, the set of slides where $u$ is incident to the top corner of $f$ must be a contiguous subsequence of $\Omega'_v$ that we identify in $O(\log n)$ time.
\item Here, we claim that $u \tor v$ is conflicted for $F_s = (f, w, v)$ if and only if it is conflicted for all of $F_s' = (f', w, v)$. 
Indeed, consider the two minimal subwalks $\pi_b^*$ and $\pi_b'$ from $v$ to $w$ along the face $g$ (one clockwise along $g$, the other counterclockwise).
After all slides $F_s' = (f', w, v)$, the minimal subwalk from $v$ to $u$ must include $\pi_b^*$ (or they all must include $\pi_b'$). 
Assume $\pi_b^*$.
We showed in Lemma~\ref{lemma:legal_arti} that these subwalks must end in edges $e_1$ or $e_1' $ directed away from $w$.
Thus, $v \tor u$ cannot be conflicted in $f'$ because of a directed path from $v$ to $u$.
It follows that $v \tor u$ can only be conflicted in $f'$ because $\pi_b^*$ includes a sink corner before a source corner.
However, by Observation~\ref{obs:shortestpath}, then $v \tor u$ must be conflicted for all slides $F_s'$.
We identify $\pi_b^*$ and test if there is a sink corner before a source corner in $O(\log n)$ time, anologue to Lemma~\ref{lemma:uplinkable_query}.
    \item After a slide $F_s$, $f$ is incident to an edge directed towards $v$ if and only if $g$ was incident to an edge directed towards $v$.
We test this in $O(\log n)$ time by checking the radial ordering of edges around $v$.
\end{enumerate}
The above analysis shows that we can identify the set $\Omega^*_v$ in $O(\log n)$ time.

To find the set $\Omega^*_u$ of slides $F^u_s  = (f, w, u)$ that are part of our output, we do a near-identical procedure with a few significant changes: 
given the set $\Omega'_u$, we want to find the subset $\Omega^*_u$ of slides that are part of our output. 
By the above analysis, if $\Omega_u'$ is not empty, then $C_2$ cannot contain a top corner on its outer cycle and thus $u$ cannot be incident to a top corner. 
Moreover for any slide $F_s^u = (f, w, u) \in \Omega'_v$, the articulation vertex $w$ cannot be incident to the top corner of $f$. 
It follows from the conditions of Lemma~\ref{lemma:conflict} that after any slide $F_s^u \in \Omega'_v$, $u \tor v$ may be inserted without violating upward planarity if and only if $(w, v)$ may be inserted without violating upward planarity. 
By the proof of Lemma~\ref{lem:uplinkable}, we can identify this subset $\Omega^*_u$ in $O(\log^2 n)$ additional time. 

What remains is to verify whether $v < u$. 
If $\Omega^*_u \cup \Omega^*_v$ are not empty, we simply perform the first articulation slide and use Lemma~\ref{lem:uplinkable} to check condition $(iv)$ (and thus whether $v < u$).
We report either all or none of $\Omega^*_u \cup \Omega^*_v$ accordingly.
\end{proof}

\mypar{Twist-UpLinkable$(u, v)$.}
With the above analysis, we can almost immediately show the following:

\begin{lemma}
We support the Twist-UpLinkable$(u, v)$ query in $O(\log^2 n + k)$ time.
\end{lemma}

\begin{proof}
Consider the set $\Lambda_v$ of all articulation twists  $F_t = (w, v)$ where $u$ does not share a component with $v$ in the graph $G \backslash \{ w \}$.
Observe that this set is empty whenever $u$ and $v$ share more than one face (since then, $u$ and $v$ must be biconnected in $G$). 
We use Linkable$(u, v)$ in $O(\log^2 n)$ time to obtain the faces shared by $u$ and $v$ and we test in $O(1)$ if this outputs exactly one face $f$. 
If it does not, we conclude that our output is empty.

Given the unique face $f$, the set $\Lambda_v$ corresponds to all articulation points $w$ on the path from $u$ to $v$, which have two corners incident to $f$. This set is \emph{exactly} what the Linkability algorithm by Holm and Rotenberg~\cite{holm2017dynamic} identifies and thus we obtain $\Lambda_v$ in $O(\log^2 n)$ time.

We now note that articulation twists never violate update planarity: 
For all $F_t = (w, v)$, the component separated by $w$ must have two outgoing edges incident to $f$ (otherwise, the component must contain the unique source of $G$ on its border). 
It follows that any twist $F_t = (w, v)$ does not alter the corners incident to $w$ and thus, combinatorially, does not change the face sink graph. 
We conclude our argument by showing that $u \tor v$ is uplinkable after either all twists in $\Lambda_v$, or no twist in $\Lambda_v$. 
Indeed, consider the conditions of Lemma~\ref{lem:uplinkable}: 
Conditions $(i)$, $(iii)$ and $(iv)$ remain unchanged after a twist in $\Lambda_v$.
So $u \tor v$ is not uplinkable after $\Lambda_v$ if and only if $u \tor v$ is conflicted in $f$. 

Denote by $F^* = (w^*, v)$ the first articulation twist in $\Lambda_v$. 
Consider the two minimal subwalks $\pi_b^*$ and $\pi_b'$ from $v$ to $w^*$ along the face $f$ (one clockwise along $g$, the other counterclockwise).
After any twist in $\Lambda_v$, the minimal subwalk from $v$ to $u$ along $f$ will include $\pi_b^*$ (or $\pi_b'$).
Identical to Lemma~\ref{lemma:slideuplinkable}, by Observation~\ref{obs:shortestpath}, this implies that $u \tor v$ is conflicted after any flip if and only if $\pi_b^*$ contains a source before a sink.
We can test this using Theorem~\ref{thm:dynamicdatastructure} and output the result accordingly.

Finally, we remark that we can do an identical procedure for the twists $F_t^u = (w, u)$ in the corresponding set $\Lambda_u$.
\end{proof}

\newpage

\mypar{Separation-UpLinkable$(u, v)$} reports at most one separation flip $F^* = (f, g, x, y, u / v)$ after which the edge $u \tor v$ may be inserted across a face in $\E(G)$. 
First we define and identify the set $\Sigma_v$ of separation flips $F = (f_u, g_v, x, y, v)$.
Recall that $(x,y)$ is a separation pair such that cutting $(x, y)$ along the corners of $f_u$ and $g_v$, separates $u$ and $v$ in $G$. 
Moreover, $v$ is incident to $g_v$ before the operation, and shares $f_u$ with $u$ afterwards. 
The following observation follows immediately from the definition of $F$:

\begin{observation}
\label{obs:normalcase}
For any $F = (f_u, g_v, x, y, v)$ and $F' = (f_u' , g_v' , x', y', v)$ in $\Sigma_v$, $g_v = g_v'$.
\end{observation}

\noindent
We first consider the special case where $\Sigma_v$ contains at least two elements $F$ and $F'  = (f_u', g_u', x', y', v)$ with $f_u \neq f_u'$ and we observe the following:

\begin{observation}
\label{obs:specialcase}
If $\Sigma_v$ contains at least two elements $F = (f_u, g_v, x, y, v)$ and $F' = (f_u' , g_v' , x', y', v)$ with $f_u \neq f_u'$ then:
$(x, y) = (x', y')$ and  $|\Sigma_v| = 2$.
\end{observation}

\noindent
Whenever, at query time, we identify $\Sigma_v$ and detect that $|\Sigma_v| \leq 2$, we simply try all elements of $\Sigma_v$ first applying the separation flip, and then testing for uplinkability using Lemma~\ref{lemma:uplinkable_query}. 
Henceforth, we assume that $|\Sigma_v| > 2$.
Thus, for all $F$ and $F'$ in $\Sigma_v$, the face $f_u = f_u'$ and $g_v = g_v'$. 
Next, we consider the set $\Sigma_v' \subseteq \Sigma_v$ of separation flips that do not violate upward planarity. 
We show that $\Sigma_v'$ is not empty if and only if it contains a flip $F^*$ where $F^*$ has the closest separation pair $(x^*, y^*)$ to $v$. 
Finally, we show that if there exist separation flip in $\Sigma_v$ where, after $F$, $u \tor v$ are uplinkable in $f_u$ then this must be true for $F^*$. We test the validity of $F^*$ using Theorem~\ref{thm:dynamicdatastructure} and Lemma~\ref{lemma:uplinkable_query}.

\begin{lemma}
\label{lemma:find_sep-flips}
We can identify from the set $\Sigma_v$ in $O(\log^2 n)$ an element $F^* = (f_u, g_v, x^*, y^*) \in \Sigma_v$ where $(x^*, y^*)$ is closest to $v$ in $O(\log n)$ additional time.
\end{lemma}

\begin{proof}
This follows almost immediately from the proofs of Theorem 5 in \cite{holm2017dynamic}. However, to formally show this, we slightly open the black box to investigate their argument.
\begin{itemize}[noitemsep, nolistsep]
    \item If we are in the special case as by Observation~\ref{obs:specialcase}, their proof gives us exactly  the separation pair $(x, y)$ and the at most four flips after which $u$ and $v$ share a face (where at most two of these separation flips are in $\Sigma_v$).
    \item If we are not in the special case, and the set $\Omega_v$ of articulation slides involving $v$ and after which $u$ and $v$ share a face is empty, then we can immediately apply Theorem 5 in \cite{holm2017dynamic}.
    In \cite{holm2017dynamic}, Holm and Rotenberg use the fact that $\Omega_v$ is empty to identify the face $g_v$ corresponding to all separation flips in $\Sigma_v$. 
    Given $g_v$, they identify an ordered set of vertices $\Sigma$ on a fundamental cycle between $u$ and $v$ where: every pair of vertices $(x, y) \in \Sigma \times \Sigma$ (with the fundamental cycle traversing $(x, u, y, v)$ in this order) forms a separation pair that separates $u$ and $v$. 
    Since this is an ordered sequence of vertices, we obtain the closest pair $(x^*, y^*)$ in $O(\log n)$ additional time.
    \item If we are not in the special case and the set $\Omega_v$ is not empty, we cannot immediately apply their proof. Instead we identify the face $g$ incident to all articulation slides in~$\Omega_v$. 
    Now $g$ must be equal to the face $g_v$ corresponding to the separation flips in $\Sigma_v$. 
    Since we now have identified the face $g_v$ of the aforementioned proof, we again obtain the set $\Sigma \times \Sigma$, and the pair $(x^*, y^*)$. \qedhere
    \end{itemize}
\end{proof}

\begin{lemma}
\label{lemma:only_sep_you_need}
Let $G$ be a single-source digraph and $\E(G)$ be an upward planar embedding.
Let $|\Sigma_v| > 2$ and $F^* = (f_u, g_v, x^*, y^*) \in \Sigma_v$ with $(x^*, y^*)$ closest to $v$.
If $F^*$ violates upward planarity, then all separation flips in $\Sigma_v$ violate upward planarity. 
\end{lemma}

\begin{proof}
Since $|\Sigma_v| > 2$ it must be that all separation flips in $\Sigma_v$ involve the same faces $f_u$ and $g_v$ (Observation~\ref{obs:specialcase}).
Per definition, $u$ is incident to $f_u$ and $v$ is incident to $g_v$ before the flip.
We denote by $F' \in \Sigma_v$ any flip in $\Sigma_v$ which is not $F^*$.
By Observations~\ref{obs:normalcase} and \ref{obs:specialcase}, $F' = (f_u, g_v, x', y', v)$ for some separation pair $(x', y')$. 

\mypar{The proof setting.}
We first consider the proof setting (Figure~\ref{fig:newsink_setting}). 
Let $(x^*, y^*)$ separate $G$ in a flip component $C_1$ containing $u$ and a flip component $C_2$ containing $v$. 
We denote by $\pi$ the walk bounding $C_2$.
By $\pi_v$ we denote the subwalk of $\pi$ incident to $g_v$ and by $\pi_u$ its complement (incident to $f_u$). 
For the alterate flip $F'$, we define $\pi'$, $\pi_v'$, and $\pi_u'$ anologously.
Since $(x^*, y^*)$ is the closest pair to $v$, it must be that $\pi_v \subseteq \pi_v'$ and $\pi_u \subseteq \pi_v'$. 

It follows that the flip $F^*$ or alternative flip $F'$ cause $\pi_u$ to become incident to $g_v$ and $\pi_v$ to become incident to $f_u$. 
We note that both $F^*$ and $F'$ immediately violate upward planarity whenever $\pi_u$ or $\pi_v$ contain in their interior:
 the unique sink of $G$ (as this sink must remain incident to the outer face), or
 the top corner of $f_u$, respectively $g_v$ (since whenever the path contains the top corner, it must also contain a source in its interior. However, the digraph has one unique source which must remain incident to the outer face). 
Moreover, there can be no sink corners incident to $x^*$. Indeed, the first edge $e_v^*$ on $\pi_v$ and the first edge $e_u^*$ on $\pi_u$ must both be outgoing edges (else, one of $\pi_v$ or $\pi_u$ must contain the unique source).  Thus, there are never any sink corners of $f_u$ or $g_v$ incident to $x^*$.

\begin{figure}[hb]
\centering
\includegraphics[]{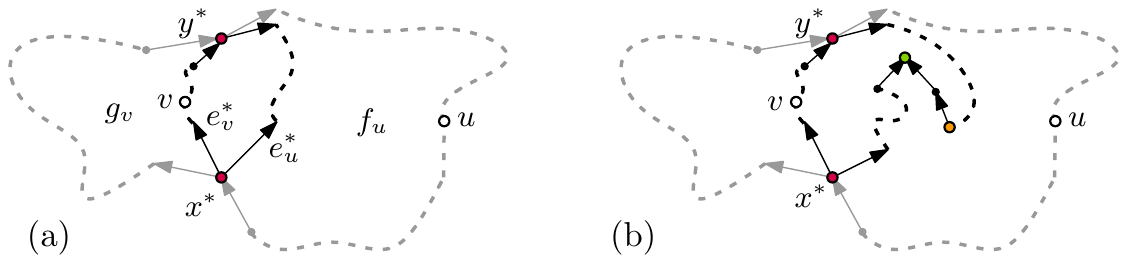}
\caption{
(a) A $u$ and $v$ in white, $(x^*, y^*)$ in red, and the faces $f_u$ and $g_v$. The paths $\pi_u$ and $\pi_v$ are shown in black.
The edges $e_u^*$ and $e_v^*$ must point outwards of $x^*$.
(b)  If $\pi_u$ contains the top corner of $f_u$ in its interior (green) then $\pi_u$ must contain a sink of $G$ (orange).
}
\label{fig:newsink_setting}
\end{figure}

\mypar{Assuming $F^*$ violates upward planarity.}
After $F^*$, all children of $f_u$ in the face-sink graph that lie on  $\pi_u$, become children of $g_v$. Similarly, all children of $g_v$ on $\pi_v$, become children of $f_u$. 
Since $f_u$ and $g_v$ keep their respective parents in the face sink graph, this operation cannot invalidate the face-sink graph (i.e., all these displaced children remain part of a valid tree in $\mathcal{F}(\E(G))$).

So suppose that after $F^*$, the face-sink graph is invalid. Then it must be that $F^*$ either adds or removes at least sink corner to $\E(G)$ (either disconnecting a tree in the face-sink graph, or adding a new invalid tree).
From the above analysis, this new or removed sink corner must be incident to $y^*$. 
Denote by $e_u$ and $e_v$ the last edges of $\pi_u$ and $\pi_v$, respectively. 
Denote by $e'_u$ (and $e'_v$) the other edges incident to both $y^*$ and $f_u$ (or $g_v)$.
If $F^*$ changes the number of sink corners incident to $y^*$ then 
$e_u$ is directed towards $y^*$ whilst $e_v$ is directed from $y^*$ (or vice versa).
Indeed, if they are both outgoing or incoming edges then the number of sink corners incident to $y^*$ must remain the same after $F^*$.
At least $e'_u$ or $e'_v$ must be directed towards $y^*$. 
Suppose that after $F^*$, $\mathcal{F}(\E(G))$ loses a sink corner $c^*$.
We now make a case distinction
on where whether one of $e_u'$ and $e_v'$ is not directed towards $y^*$ (Figure~\ref{fig:newsink_cases}(a + b))
or both are (Figure~\ref{fig:newsink_cases}(c + d)).

\mypar{The high-level argument.}
In each of these four cases, we will show that $F^*$ violates upward planarity when $\mathcal{F}(\E(G))$ loses a sink corner incident to $y^*$.
What remains, is to show that $F' = (f_u, g_v, x', y', v)$ also must violate upward planarity. 
If $y' = y^*$, then $F'$ must also violate upward planarity because it loses a sink corner incident to $y^*$. 
So assume that $y'$ succeeds $y$ on a path from $v$ to $u$. 
Per definition of the separation flip, $y'$ must be incident to both $f_u$ and $g_v$.
For the sake of contradiction, we assume that $F'$ does not violate upward planarity.
We show that this implies that the walk $\pi_v'$ must include a vertex $z$ which (on $\pi_v'$) has two outgoing edges. Since $\pi_v'$ cannot be the source of $G$, there must be a path  $\pi_1$ from the source of $G$ to $z$.
Similarly, there must be a path $\pi_2$ from the source of $G$ to $x^*$. 
However, the paths $\pi_1, \pi_2$ and $\pi_v$ now enclose $g_v$ (Figure~\ref{fig:newsink_cases}); this contradicts that $y'$ is incident to $g_v$. 
We are now ready for case analysis for when the flip $F^*$ removes a corner from $\mathcal{F}(\E(G))$:

\textbf{\boldmath Case (a): $e_u$ and $e_u'$ are directed towards $y^*$ and $e_v$ and $e_v'$ are not. }
The sink corner between $e_u$ and $e_u'$ must be a top corner.
Indeed, if the angle between $e_u$ and $e_u'$ is reflex, the edges $e_u'$ and $e_v'$ must be pointing downwards.
After performing $F^*$, $f_u$ is no longer incident to a top corner. This makes $f_u$ the root of a tree in the face-sink graph. The only roots of the face-sink graph can be a critical vertex.  $f_u$ is not the outer face because it was incident to a top corner. Thus, $F^*$ violates upward planarity. 
        
What remains is to show that $F'$ violates upward planarity. 
Assuming $F'$ does not violate upward planarity implies $F'$ makes the top corner of $f_u$ incident to $g_v$ instead. 
This in turn implies that $\pi_v'$ contains the top corner of $g_v$. Thus, $\pi_v'$ must contain a vertex $z$ where the edges of $\pi_v'$ point outwards of $z$.
Via the above argument, this is a contradiction.

\textbf{\boldmath Case (b):  $e_v$ and $e_v'$ are directed towards $y^*$ and $e_u$ and $e_u'$ are not.}
The flip $F^*$ violates upward planarity via the same reasoning as in case (a) where now $g_v$ is without a top corner.
So suppose for the sake of contradiction that $F'$ does not violate upward planarity.
Since $e_v'$ is directed towards $y^*$, subpath of $\pi_v'$ from $y^*$ to $y'$ must include edges which point upwards. 
This implies that $\pi_v'$ must include upward pointing edges and thus at least one vertex $z$ with two edges in $\pi_v'$ that are directed outwards of $z$.
Just as in case~(a), this implies that $g_v$ cannot be incident to $y'$ which is a contradiction. 

\textbf{\boldmath Case (c): $e_u, e_u', e_v'$ are directed towards $y^*$ and $e_v$ is not.}
The corner incident to $e_u$ and $e_u'$ must be the top corner of $f_u$: performing the flip would make $f_u$ an invalid root of a tree in the face-sink graph. 
If $F'$ does not violate upward planarity, $\pi_v'$ must contain the top corner of $g_v$ which introduces the same contradiction as case (a). 

\textbf{\boldmath Case (d): $e_v, e_u', e_v'$ are directed towards $y^*$ and $e_u$ is not.}
The corner incident to $e_v$ and $e_v'$ must be the top corner of $g_v$: performing the flip would make $g_v$ an invalid root of a tree in the face-sink graph. 
For any alternate flip $F'$, $\pi_v'$ must contain additional upward edges to reach $y'$, which implies the existence of at least one vertex $z$ with two edges in $\pi_v'$ that are directed outwards of $z$. 

\mypar{Concluding the argument.}
We showed that if $F^*$ causes $\mathcal{F}(\E(G))$ to lose a sink corner, then $F^*$ and any alternative flip $F'$ must violate upward planarity. 
We claim this implies that  
$F^*$ can never cause $\mathcal{F}(\E(G))$ to gain a sink corner. 
Indeed, suppose for the sake of contradiction that $F^*$ causes $\mathcal{F}(\E(G))$ to gain a sink corner without violating upward planarity.
Then undoing $F^*$ is a separation flip in an upward planar combinatorial embedding $\E(G)'$ that causes $\mathcal{F}(\E(G)')$ to lose a sink corner. However, our above analysis shows that the resulting embedding (which is the original $\E(G)$) then is not upward planar which is a contradiction. 
This concludes our argument.
\begin{figure}[t]
\centering
\includegraphics[page = 2]{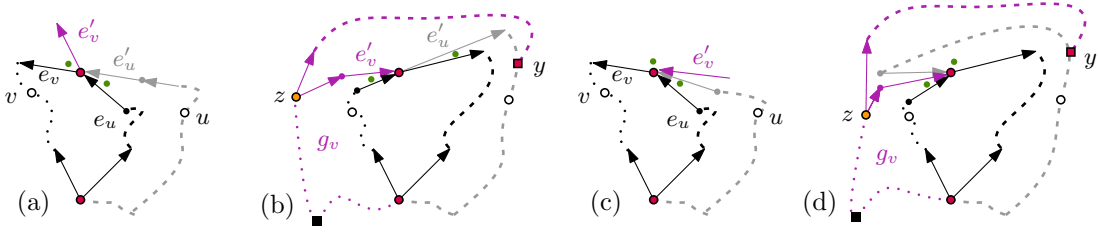}
\caption{
The vertices $u$ and $v$ are white, $(x^*, y^*)$ are round and red. 
We show in all four cases that $F^*$ around $(x^*, y^*)$ violates upward planarity.
For any other separation flip $F'$ including $y$, the purple walk must include a vertex $z$ with two outgoing edges.
}
\label{fig:newsink_cases}
\end{figure}
\end{proof}

\begin{lemma}
\label{lemma:really_only_one}
Let $G$ be a single-source digraph and $\E(G)$ be an upward planar embedding.
Let $(v, u)$ be not uplinkable in $G$ and let $|\Sigma_v| > 2$.
Let $(x^*, y^*)$ be the pair of vertices, closest to $v$, that correspond to a flip $F^* = (f_u, g_v, x, y, v)$ in $\Sigma_v$. 
If after $F^*$,  $u \tor v$ (or $v \tor u$) is not uplinkable, then after all separation flips in $\Sigma_v$ $u \tor v$ (or $v \tor u$) are not uplinkable.
\end{lemma}

\begin{proof}
Let $F = (f_u', g_v', x, y, v)$ be a separation flip in $\Sigma_v$ which does not violate  upward planarity and $F \neq F^*$.
If $|\Sigma_v| > 2$ then by Observation~\ref{obs:normalcase} and \ref{obs:specialcase}: $(f_u, g_v) = (f_u', g_v' )$.
By Lemma~\ref{lem:uplinkable}, the edge $u \tor v$ may be inserted across $f_u$ if and only if:
\begin{enumerate}[$(i)$, noitemsep, nolistsep]
    \item $u$ is not incident to the top corner of $f_u$,
    \item  $u \tor v$ is  not conflicted in $f_u$, and
    \item $f_u$ is incident to an edge directed towards $v$. 
\end{enumerate}
Now note that:
\begin{enumerate}[$(i)$, noitemsep, nolistsep]
    \item the top corner of $f_u$ must lie in the flip component containing $u$ for both $F$ and $F^*$,
    \item 
    Denote by $C^*$ the component that is mirrored by the separation flip $(x^*, y^*)$ i.e. the component containing $v$.
    Denote for any alternate flip $F' = (f_u, g_v, x, y, v)$ by $C'$ its component mirrorred by the separation flip.
    Per definition of minimal, $C^* \subset C$. 
    
    Suppose that $v \to u$ is in conflict after $F^*$ because $f_u$ then has a directed path towards $u$. Then this must also be true for $F'$. 
    Otherwise, just as in Lemma~\ref{lemma:slideuplinkable}, $v \to u$ must be in conflict after $F^*$ only due to a sink corner in $C^*$.
    Specifically, the shortest path $\pi_b^*$ from $v$ to $u$ along $f$ must include a sink corner before a source corner, and this sink corner must lie in $C^*$ (Observation~\ref{obs:shortestpath}). 
    It immediately follows from $C^* \subset C$ that for all other separation flips $F'$, $u \to v$ must also be conflicted in $f$.
    \item all edges incident to $f_u$ and $v$ must lie in $C^*$ and the argument follows.
\end{enumerate}
We proved that $u \tor v$ may be inserted after $F$, if and only if it may be inserted after $F^*$.

Next, we show the same for the edge $v \tor u$.
\begin{enumerate}[$(i)$, noitemsep, nolistsep]
    \item By the proof of Lemma~\ref{lemma:only_sep_you_need}, the vertex $v$ cannot be incident to the top corner of $g_v$ before the flip (and $f_u$ after the flip).
    \item If $v \tor v$ is conflicted after $F^*$, they can either be conflicted due to a directed path incident to $f$ (which remains immutable) or due to a sink corner $c_s$ in the graph $G \backslash C^*$. 
    Any other flip $F' = (f_u, g_v, x, y, v)$ must keep this sink corner $c_s$ intact unless $u \in (x, y)$. 
    \item If $u \not \in (x, y)$ then $F'$ does not change the edges incident to $u$.  It follows that if after $F$ we violate Condition $(iii)$ then the same is true for $F'$ unless $u \in (x, y)$.
\end{enumerate}
It immediately follows that if $u \not \in (x, y)$ then $v \tor u$ may be inserted after $F$, if and only if it may be inserted after $F^*$.

Now suppose for the sake of contradiction that $u \in (x, y)$. Moreover, suppose that $v \tor u$ may be inserted after $F$. 
The definition of separation flip $u \in (x, y)$ implies that before $F$, $u$ also shares $g_v$ with the vertex $v$.
Moreover, if $v \tor u$ may be inserted across $f_u$ after $F$ then by Lemma~\ref{lem:uplinkable}, $v \tor u$ may be inserted across $g_v$ before $F$. However, this contradicts our assumption that before the flip, $v \tor u$ was not uplinkable. 
\end{proof}

\begin{lemma}
Let $u \tor v$ be not uplinkable in the current embedding $\E(G)$. Then we can answer Separation-UpLinkable$(u, v)$ in $O(\log^2 n)$ time.
\end{lemma}

\begin{proof}
We test whether there exists a separation flip $F = (f_u, g_v, x, y, v) \in \Sigma_v$ such that after the separation flip $u \tor v$ become uplinkable.
By Lemma~\ref{lemma:find_sep-flips}, we can find the separation flip $F^* \in \Sigma_v$ closest to $v$ that flips $v$ into a face $f_u$ incident to $u$. 
We use Theorem~\ref{thm:dynamicdatastructure} test whether performing $F^*$ would violate upward planarity in $O(\log^2 n)$ time (i.e. perform the separation flip, check if it got rejected). 
If $F^*$ violates upward planarity, we know by Lemma~\ref{lemma:only_sep_you_need} that there can be no separation flip involving $v$ (but not $u$) that does not violate upward planarity. 
If $F^*$ does not violate upward planarity, we check if the conditions of Lemma~\ref{lem:uplinkable} are met in $O(\log n)$ additional time.
If so, we report $F^*$. Otherwise, by Lemma~\ref{lemma:really_only_one} we know that there cannot be a separation flip in $\Sigma_v$ after which $u \tor v$ becomes uplinkable and we terminate our search.

Then, we test whether there exists a separation flip in the symmetrically-defined $\Sigma_u$ such that after the separation flip $u \tor v$ become uplinkable.
We do this in an identical way as above.
Observe that for this case, our assumption that $u \tor v$ is not uplinkable in $\E(G)$ allows us to apply Lemma~\ref{lemma:really_only_one}. This concludes the proof.
\end{proof}

\mypar{One-Flip-UpLinkable$(u, v)$:} can now immediately be done in $O(\log^2 n)$ by the above analysis. 
Indeed, by Holm and Rotenberg~\cite{holm2017dynamic} and our observations, we identify the unique projected articulation points $\pi(v)$ and $\pi(u)$ of $\Omega_v$ and $\Omega_u$ (if any).
Using Linkable$(\pi(u), \pi(v))$ we find all faces $\Delta$ shared between $\pi(v)$ and $\pi(v)$. We then filter from $\Delta$, in an identical manner as above, the faces $f$ where sliding and/or twisting $u$ and $v$ into $f$ makes $u \tor v$ uplinkable. 
For including separation flips, we find the unique faces $f_u$ and $f_v$ of $\Sigma_v$ and do the same consideration for slides, twists, or the separation flip that includes $u$.

\section{Extended preliminaries: Top trees.}
\label{ap:toptrees}
The face-sink graph is a collection of trees $\T^*, \T_1, \ldots, \T_m$.
We observe that every planar embedding $\E(G)$ induces a planar embedding of each tree $\T$. 
For each of these trees $\T$, we will maintain an embedding-sensitive decomposition of $\T$ which is known as a \emph{embedding-sensitive top tree} $\e(\T)$~\cite{alstrup2005maintaining, tarjan2005self}. 
Formally, for any connected subgraph $S$ of $\T$, we define the \emph{boundary vertices} of $S$ as the vertices in $S$ that are incident to an edge in $\T\setminus S$.  A \emph{cluster} is then a connected subgraph of $\T$ with at most $2$ boundary vertices.
We say that a cluster with one boundary vertex is a \emph{point cluster}, and with two boundary vertices is a \emph{path cluster}.
A top tree $\e(\T)$ is a hierarchical decomposition of $\T$ (with depth $O(\log n)$) into point and path clusters that is structured as follows: 
the leaves of $\e(\T)$ are the path and point clusters for each edge $u \tor v$ in $\T$ (where a leaf in $e(\T)$ is a point cluster if and only if the corresponding edge $u \tor v$ is a leaf in $\T$).
Each inner node $\nu \in \e(\T)$ \emph{merges} a constant number of \emph{child} clusters sharing a single vertex, into a new point or path cluster. The vertex set of $\nu$ merges those corresponding to its children. 
Furthermore, embedding-respecting top trees only allow merges of neighbouring clusters according to the edge-ordering of the embedding of $\T$ (to obtain such an embedding, simply fix a left-to-right ordering of all children of nodes in $\T$)  
Such embedding-respecting top trees have the following properties~\cite{holm2017dynamic}:
\begin{property}
\label{prop:toptree}
Any embedding-respecting top trees over trees with at most $n$ vertices support the following operations in $O(\log n)$ time:
\begin{itemize}[noitemsep, nolistsep]
    \item \textbf{\boldmath Insert/delete$(\e(T), \T, \T')$} updates $\e(T)$ after a insertion or deletion on $\T$.
    \item \textbf{\boldmath Separate$(\e(\T), v, v_1, v_2)$} constructs the top trees $\e(\T_1)$ $\e(\T_2)$ of the trees $\T_1$ and $\T_2$ that are obtained by separating $v$ into two nodes $v'$ and $v^*$ where $v'$ receives the children of $v$ from $v_1$ to $v_2$ and $v^*$ receives the complement (this is equal to splitting $v$ along two corners in the embedding). 
    
        \item \textbf{\boldmath Join$(\e(\T_1), \e(\T_2), u, v)$} constructs the top tree $\e(\T)$ from the tree  $\T$ that is obtained by joining $\T_1$ and $T_2$ through vertices $u \in \T_1$ and $v \in \T_2$.
        \item \textbf{\boldmath Invert$(\e(\T), v, u, w)$} updates $\e(\T)$ after selecting for a vertex $v$, a subsequence of children $(u, \ldots w)$ and inverting their left-to-right order.
\end{itemize}
\end{property}
 
\bibliographystyle{plain}
\bibliography{refs}

\end{document}